\pdfoutput=1

\documentclass[letterpaper,11pt]{article}

\usepackage{amsmath}
\usepackage{amsfonts}
\usepackage{amssymb}
\usepackage{graphicx}
\usepackage{wrapfig}
\usepackage[linesnumbered,algoruled,vlined]{algorithm2e}
\usepackage{xspace}
\usepackage[margin=1.2in]{geometry}
\usepackage{paralist}

\usepackage{color}
\usepackage{soul}

\newcommand{\namedref}[2]{\hyperref[#2]{#1~\ref*{#2}}}
\newcommand{\sectionref}[1]{\namedref{Section}{#1}}

\newcommand{\theoremref}[1]{\namedref{Theorem}{#1}}

\newcommand{\figref}[1]{\namedref{Figure}{#1}}
\newcommand{\lemmaref}[1]{\namedref{Lemma}{#1}}

\newcommand{\corollaryref}[1]{\namedref{Corollary}{#1}}

\newcommand{\algref}[1]{\namedref{Algorithm}{#1}}
\newcommand{\lineref}[1]{\namedref{Line}{#1}}
\newcommand{\equalityref}[1]{\hyperref[#1]{Equality~(\ref*{#1})}}
\newcommand{\inequalityref}[1]{\hyperref[#1]{Inequality~(\ref*{#1})}}
\newcommand{\pprtyref}[1]{\hyperref[#1]{Property~(\ref*{#1})}}
\newcommand{\factref}[1]{\namedref{Fact}{#1}}

\newtheorem{theorem}{Theorem}[section]
\newtheorem{lemma}[theorem]{Lemma}
\newtheorem{definition}[theorem]{Definition}
\newtheorem{corollary}[theorem]{Corollary}

\newtheorem{fact}{Fact}[section]

\newcommand{\blackslug}{\hbox{\hskip 1pt \vrule width 4pt height 8pt
depth 1.5pt \hskip 1pt}}
\newcommand{\QED}{\quad\blackslug\lower 8.5pt\null\par}
\newenvironment{proof}[1][Proof:]{\noindent \textbf{#1}\xspace}{\QED}

\newcommand{\N}{\mathbb{N}}
\newcommand{\R}{\mathbb{R}}
\newcommand{\E}{\mathbb{E}}
\newcommand{\BO}{\mathcal{O}}

\newcommand{\Set}[1]{\left\{ #1 \right\}}
\newcommand{\ceil}[1]{\left\lceil #1 \right\rceil}
\newcommand{\Seq}[1]{\left\langle #1 \right\rangle}
\newcommand{\ignore}[1]{}
\newcommand{\DEF}{\ensuremath{\stackrel{\rm def}{=}}}
\newcommand{\REM}[1]{\hfill//\emph{#1}}

\newcommand{\Src}{\mathrm{source}}

\newenvironment{eqntext}{\[\begin{array}{rcl@{\hspace{1cm}}l}}{\end{array}\]}

\def\midformat{
\setlength{\textheight}{8.9in}
\setlength{\textwidth}{6.7in}
\setlength{\evensidemargin}{-0.2in}
\setlength{\oddsidemargin}{-0.2in}
\setlength{\headheight}{0in}
\setlength{\headsep}{10pt}
\setlength{\topsep}{0in}
\setlength{\topmargin}{0.0in}
\setlength{\itemsep}{0in}       
\renewcommand{\baselinestretch}{1.1}
\parskip=0.070in
}

\newcommand{\CONGEST}{\textbf{CONGEST}}
\newcommand{\BSP}{\mathrm{BSP}\xspace}
\newcommand{\Next}{\mathrm{next}}
\newcommand{\Paths}{\mathrm{paths}}
\newcommand{\Hd}{\mathrm{hd}}
\newcommand{\HD}{\mathrm{HD}}
\newcommand{\Wd}{\mathrm{wd}}
\newcommand{\WD}{\mathrm{WD}}
\newcommand{\SPD}{\mathrm{SPD}}

\newcommand{\Ball}{\mathit{ball}}

\newcommand{\Lead}{Y}
\newcommand{\gsf}{\textsc{gsf}}
\newcommand{\OPT}{\mathrm{OPT}}
\newcommand{\ALG}{\mathrm{ALG}}

\newcommand{\cI}{\mathcal{I}}
\newcommand{\rtc}{\textsc{rtc}}

\newcommand{\cM}{\mathcal{M}}

\midformat

\hypersetup{colorlinks,linkcolor=blue,filecolor=blue,citecolor=blue,urlcolor=blue,pdfstartview=FitH}

\begin{document}
\setcounter{tocdepth}{3}
\date{}

\title{\textbf{Fast Routing Table Construction Using Small Messages}}

\author{
Christoph Lenzen%
\thanks{Supported by the Swiss Society of Friends of the Weizmann
Institute of Science and by the Swiss National Science Foundation (SNSF).}
\\
Dept.\ Computer Science \& Applied Mathematics\\
Weizmann Institute of Science\\
Rehovot 76100, Israel
\and
Boaz Patt-Shamir%
\thanks{Supported in part by the Israel Science
Foundation (grant 1372/09) and by Israel Ministry of Science and
Technology.}
\\
School of Electrical Engineering\\
Tel Aviv University\\
Tel Aviv 69978, Israel
}

\def\thepage{}
\begin{titlepage}
  
\maketitle

\begin{abstract}
We describe a distributed randomized algorithm computing approximate distances
and routes that approximate shortest paths. Let $n$ denote the number of
nodes in the graph, and let $\HD$ denote the \emph{hop diameter} of the graph,
i.e., the diameter of the graph when all edges are considered to have unit
weight. Given $0<\varepsilon\leq 1/2$, our algorithm runs in $\tilde
\BO(n^{1/2+\varepsilon}+\HD)$ communication rounds using messages of $\BO(\log
n)$ bits and guarantees a stretch of $\BO(\varepsilon^{-1}\log
\varepsilon^{-1})$ with high probability. This is the first distributed
algorithm approximating weighted shortest paths that uses small messages and
runs in $\tilde o(n)$ time (in graphs where $\HD\in \tilde{o}(n)$). The time
complexity nearly matches the lower bounds of $\tilde\Omega(\sqrt n+\HD)$ in the
small-messages model that hold for \emph{stateless} routing (where routing
decisions do not depend on the traversed path) as well as approximation of the
weigthed diameter. Our scheme replaces the original identifiers of the nodes by
labels of size $\BO(\log \varepsilon^{-1}\log n)$. We show that no algorithm
that keeps the original identifiers and runs for $\tilde{o}(n)$ rounds can
achieve a polylogarithmic approximation ratio.

Variations of our techniques yield a number of fast distributed approximation
algorithms solving related problems using small messages. Specifically, we
present algorithms that run in $\tilde{\BO}({n}^{1/2+\varepsilon}+\HD)$ rounds
for a given $0<\varepsilon\leq 1/2$, and solve, with high probability, the
following problems:
\begin{compactitem}
\item $\BO(\varepsilon^{-1})$-approximation for the Generalized
  Steiner Forest (the running time in this case has an additive
  $\tilde\BO(t^{1+2\varepsilon})$ term, where $t$ is the number of terminals);
\item $\BO(\varepsilon^{-2})$-approximation of weighted distances, using
  node labels of size $\BO(\varepsilon^{-1}\log n)$ and
  $\tilde{\BO}(n^{\varepsilon})$ bits of memory per node;
\item $\BO(\varepsilon^{-1})$-approximation of the weighted diameter;
\item $\BO(\varepsilon^{-3})$-approximate shortest paths using the labels
$1,\ldots,n$.
\end{compactitem}
\end{abstract}
\end{titlepage}
\pagenumbering{arabic}

\section{Introduction}

Constructing routing tables is a central task in network operation, the
Internet being a prime example. Besides being an end goal on its own
(facilitating the transmission of information from a sender to a
receiver), efficient routing and distance approximation are critical
ingredients in a myriad of other distributed applications.

At the heart of any routing protocol lies the computation of short paths in
weighted graphs, where edge weights may reflect properties such as link cost,
delay, bandwidth, reliability etc. In the distributed setting, an additional
challenge is that the graph whose shortest paths are to be computed serves also
as the platform carrying communication between the computing nodes. The result
of this double role is an intriguing interplay between two metrics: the given
shortest paths metric and the ``natural'' communication metric of the
distributed system. The first metric is used for the definition of shortest
paths, where an edge weight represents its contribution to path lengths; the
other metric is implicit, controlling the time complexity of the distributed
computation: each edge is tagged by the time it takes a message to cross it. If
these two metrics happen to be identical, then computing weighted shortest paths
to a single destination is trivial (for the all-pairs problem, see below).
For the general case, the standard normalization is that messages cross each
link in unit time, regardless of the link weight; this assumption is motivated by
network synchronization. On the other hand, the length of the message must be
taken into account as well. More precisely, in the commonly-accepted  \CONGEST\
model of network algorithms \cite{Peleg:book}, it is assumed that all link
latencies are one unit and messages have fixed size, typically $\BO(\log n)$
bits, where $n$ denotes the number of nodes.

The classical algorithm for computing shortest path distributively is the
distributed variant of the Bellman-Ford algorithm. This algorithm is used in
many networks, ranging from local to wide area networks. The Bellman-Ford
algorithm enjoys many properties that make it an excellent distributed algorithm
(locality, simplicity, self-stabilization). However, in weighted graphs, its
time complexity, i.e., the number of parallel iterations, may be as high as
$\Omega(n)$ for a single destination. This is in sharp contrast with the
$\BO(\HD)$ time needed to compute \emph{unweighted} shortest paths to a single
destination, where $\HD$ denotes the unweighted ``hop-diameter'' of the network.
The difference between $n$ and $\HD$ can be huge; suffices to say that the
hop-diameter of the Internet is estimated to be smaller than 50. Intuitively,
the problem originates in the fact that the Bellman-Ford algorithm explores
paths in a hop-by-hop fashion, and the aforementioned superposition of metrics may
result in a  path that is weight-wise short, but consists of $\Omega(n)$ edges.
If shortest paths have at most $\SPD\in \N$ edges, then it suffices to run the
Bellman-Ford algorithm for $\SPD$ communication rounds. Indeed, the running time
of a few distributed algorithms is stated as a function of this or a similar
parameter for exactly this reason (e.g., \cite{DDP,KKMPT,KP-08}).

To the best of our knowledge, no distributed algorithm for computing
(approximate) weighted shortest paths in $o(\SPD)$ time in the \CONGEST\ model
was known to date. In this paper we present a distributed algorithm  that
computes approximate all-pairs shortest paths and distances using small
messages, in time that nearly matches the lower bound of
$\tilde{\Omega}(\sqrt{n}+\HD)$.

\subsection{Detailed Contributions}

Our main technical contribution, presented in \sectionref{sec:routing}, is an
algorithm that, using messages of size $\BO(\log n)$, constructs, for any
$0<\varepsilon\leq 1/2$, in $\tilde\BO(n^{1/2+\varepsilon}+\HD)$ rounds node
labels\footnote{We remark that our use of the term differs from the common
definition in that we distinguish between the auxiliary routing information
stored by the nodes (the tables) and the (preferrably very small) labels
replacing the original node identifiers as routing address.} of size $\BO(\log
\varepsilon^{-1}\log n)$ and routing tables of size
$\tilde\BO(n^{1/2+\varepsilon})$ facilitating routing and distance estimation
with stretch $\BO(\varepsilon^{-1}\log \varepsilon^{-1})$. We show that
assigning new labels to the nodes is unavoidable by proving that any
(randomized) algorithm achieving polylogarithmic (expected) stretch without
relabeling must run for $\tilde\Omega(n)$ rounds.
The running time of our algorithm is close to optimal, since known
results~\cite{DHKNPPW-11,Elkin-MST,PelegR-00} imply that computing such an
approximation in the \CONGEST\ model must take $\tilde{\Omega}(\sqrt{n}+\HD)$
rounds.

Our algorithm comprises two sub-algorithms that we believe to be of interest in
their own right. One is used for short-range routing (roughly, for the closest
$\sqrt n$ nodes), the other for longer distances. The short-range algorithm
constructs a hierarchy in the spirit of Thorup-Zwick distance
oracles~\cite{TZ-05}: A recursive structure of uniformly sampled ``landmarks''
is used to iteratively reduce the number of routing destinations (and routes)
that need to be learned, and repeated use of the triangle inequality shows that
the stretch is linear in the number of recursion stages. While this idea is not
new, our main challenge is to implement the algorithm using small messages; to
this end, we introduce a bootstrapping technique that, combined with a
restricted variant of Bellmann-Ford (that bounds the hop range and the number of
tracked sources), allows us to construct low-stretch routing tables for nearby
nodes.

This approach runs out of steam (i.e., exceeds our target complexity)
beyond the closest $\BO(\sqrt n)$ 
nodes, so at that point we switch to the ``long-distance'' scheme. 
The basic idea in this scheme is to pick roughly $\sqrt{n}$ random
nodes we call the \emph{skeleton} nodes, and to
compute all-to-all routing tables for them. This is achieved by simulating the
spanner construction algorithm by Baswana and Sen~\cite{baswana07}.
Again, the crux of the matter is an efficient implementation of this approach
using small messages. To this end, we first construct a spanner of a
graph defined by the skeleton nodes and shortest paths between them. Due to the
small number of skeleton nodes and the reduced number of edges (thanks to the
spanner construction), we can afford to broadcast the entire skeleton-spanner
graph, thereby making skeleton routing information common knowledge. In addition, we
can mark the corresponding paths in the original graph quickly. Here too, our
main low-level tool is the restricted Bellmann-Ford algorithm
that bounds both the range and the load.

Using variants of our techniques, in \sectionref{sec-ext}
we derive efficient
solutions to several 
related problems (all statements hold with high probability).
\begin{compactitem}
\item For the Generalized Steiner Forest (\gsf) problem 
we obtain, for any $0<\varepsilon\leq 1/2$, an
$\BO(\varepsilon^{-1})$-approxima\-tion within
$\tilde\BO((\sqrt{n}+t)^{1+\varepsilon}+\HD)$ rounds, where $t$ is the number of
terminals. This should be contrasted with the best known distributed approximation
algorithm for \gsf\ \cite{KKMPT}, which provides $\BO(\log n)$-approximation in
time $\tilde\BO(\SPD \cdot k)$, where $\SPD$ is the ``shortest paths diameter,''
namely the maximal number of hops in any shortest path, and $k$ is the number of
terminal components in the \gsf\ instance.
\item For any $k\in \N$, we obtain an
  $\tilde{\BO}((\sqrt{n})^{1+1/k}+\HD)$-time algorithm that
  constructs labels of size $\BO(k\log n)$ and local tables of size
$\tilde\BO(n^{1/(2k)})$, and produces distance estimations with stretch
$\BO(k^2)$. Compare with the recent distributed algorithm \cite{DDP} that
attains the same local space consumption at running time $\tilde\BO(\SPD \cdot
n^{1/(2k)})$ and stretch $4k-1$.
\item Given any $0<\varepsilon\leq 1/2$, we can compute an
  $\BO(\varepsilon^{-1})$-approximation of the diameter within
  $\tilde\BO(n^{1/2+\varepsilon}+\HD)$ rounds. We show that the standard
  construction yielding a lower bound $\tilde\Omega(\sqrt{n}+\HD)$ extends to
  this problem, implying that also for this special case our solution is close
  to optimal.
\item Employing a different routing mechanism for the short-range scheme, we can
assign the fixed labels of $1,\ldots,n$. This comes at the expense of a
larger stretch of $\BO(\varepsilon^{-3})$ within
$\tilde\BO(n^{1/2+\varepsilon}+\HD)$ rounds, for any $0<\varepsilon\leq 1/2$.
\end{compactitem}

\subsection{Related Work}


There are many centralized algorithms for constructing routing tables; in these
algorithms the goal is usually to minimize space without affecting the quality
of the routes too badly. We briefly discuss them later, since our focus is the
distributed model. At this point let us just comment that a na\"\i ve
implementation of a centralized algorithm in the \CONGEST\ model requires
$\Omega(|E|)$ time in the worst case, since the whole network topology has to be
collected at a single node just for computation.

Practical distributed routing table construction algorithms are usually
categorized as either ``distance vector'' or ``link state'' algorithm (see,
e.g., \cite{PetersonD:book}). Distance-vector algorithms are variants of the
Bellman-Ford algorithm \cite{Bellman,Ford-56}, whose worst-case time complexity
in the \CONGEST\ model is $\Theta(n^2)$. In link-state algorithms
\cite{MQRR80,OSPF}, each routing node collects the complete graph topology and
then solves the single-source shortest path problem locally. This approach has
$\Theta(|E|)$ time complexity. While none of these algorithms uses relabeling,
it should be noted that the Internet architecture in fact employs relabeling (IP
addresses, which are used instead of physical addresses,  encode some routing
information).

From the theoretical perspective, as mentioned above, there has not been much
progress in computing weighted shortest paths beyond the ``shortest path
diameter'' (we denote by $\SPD$) even for the single-source case: see, e.g.,
\cite{DDP} and references therein. For the unweighted case, an $\BO(n)$-time
algorithm for exact all-pairs shortest-paths was recently discovered
(independently) in \cite{HW12} and \cite{PLT-12}. These algorithms do not
relabel the nodes. In addition, a randomized $(3/2)$-approximation of $\HD$ is
given in~\cite{PLT-12}, and a deterministic $(1+\varepsilon)$-approximation is
provided by~\cite{HW12}. Combining results, \cite{HW12} and \cite{PLT-12} report
a randomized $(3/2)$-approximation of the unweighted diameter in time
$\tilde{\BO}(n^{3/4})$.

In~\cite{DHKNPPW-11}, a lower bound of $\tilde\Omega(\sqrt{n})$ on the
time to construct a shortest-paths tree of weight within a poly$(n)$
of the optimum is shown; this immediately implies the same
lower bound on routing (more precisely, on \emph{stateless} routing,
where routing decisions depend only on the destination and not on the
traversed path). To the best of our knowledge, 
the literature does not state any further explicit lower bounds on the running
time of approximate shortest paths or distance estimation
algorithms, but  a lower bound of
$\tilde\Omega(\sqrt{n})$ can be easily derived using the 
technique used in~\cite{DHKNPPW-11} (which in turn is based
on~\cite{PelegR-00}). 
In~\cite{FHW-12} it is shown that in the \CONGEST\ model,
approximating the diameter of unweighted graphs
to within a factor of $3/2-\varepsilon$ requires
$\tilde{\Omega}(\sqrt{n})$ rounds. For the unweighted case, we extend this result
to arbitrary approximation ratios.

In the Generalized Steiner Forest problem (\gsf), the input consists of a
weighted graph and a set of \emph{terminal nodes} which is partitioned into
subsets called \emph{terminal components}. The task is to find a set of edges of
minimum weight so that the terminal components are connected. Historically, the
important special case of a minimum spanning tree (all nodes are terminals,
single terminal component) has been the target of extensive research in
distributed computation. It is known that in the \CONGEST\ model, the time
complexity of computing (or approximating) an MST is
$\tilde\Omega(\sqrt{n}+\HD)$ \cite{DHKNPPW-11,Elkin-MST,PelegR-00}. This bound
is essentially matched by an exact deterministic solution~\cite{GKP93,KP98}. An
$\BO(\log n)$-approximate MST is presented in \cite{KP-08}, whose running time
is $\BO(\SPD)$, where $\SPD$ is the ``shortest path diameter'' mentioned
previously. For the special case of Steiner trees (arbitrary terminals, single
component), \cite{CF05} presents a 2-approximation algorithm whose time
complexity is $\tilde\BO(n)$ (which can easily be refined to $\tilde\BO(\SPD)$).
For the general case, \cite{KKMPT} presents an $\BO(\log n)$-approximation
algorithm whose time complexity is $\tilde\BO(\kappa\cdot\SPD)$, where $\kappa$
is the number of terminal components.%
\footnote{We note that in \cite{KKMPT}, time-optimality is claimed, up
  to factor $\tilde\BO(\kappa)$. This comes as a consequence of~\cite{KP-08},
  which in turn builds on \cite{Elkin-MST}.  However, we comment that the latter
  construction does not scale beyond the familiar lower bound of
  $\tilde\Omega(\sqrt{n})$, and a more precise statement would thus be that a
  minimum spanning tree (and thus also a \gsf) requires
  $\tilde\Omega(\min\{\SPD,\sqrt{n}\})$ rounds to be approximated.  }

We now turn to a very brief overview of centralized algorithms. 
Thorup and Zwick \cite{TZ-routing} presented an algorithm that achieves,
for any $k\in \N$, routes of stretch $2k-1$ using $\tilde\BO(n^{1/k})$
memory. 
In terms of memory consumption, it has been established that this 
scheme is optimal up to a constant factor in worst-case stretch w.r.t.\
routing~\cite{PU89}. This result has been extended to the average
stretch, and tightened to be exact up to polylogarithmic factors in memory for
the worst-case stretch~\cite{abraham06}. For distance approximation, the
Thorup-Zwick scheme is known to be optimal for $k=1,2,3,5$ and conjectured to be
optimal for all $k$ (see~\cite{zwick01} and references). The algorithm
requires relabeling with labels of size $\BO(k\log n)$.
It is unclear whether stronger lower bounds apply to name-independent routing
schemes (which keep the original node identifiers); however, for $k=1$,
i.e., exact routing, trivially $\BO(n\log n)$ bits suffice (assuming
$\BO(\log n)$-bit identifiers), and Abraham et al.~\cite{abraham08}
prove a matching upper 
bound of $\tilde\BO(\sqrt{n})$ bits for $k=2$.

A closely related concept is that of \emph{sparse spanners}, introduced
by Peleg and Sch\"affer~\cite{PS89}. A $k$-spanner of a graph is obtained by
deleting edges, without increasing the distances by more than
factor $k$. Similarly to compact routing tables, it is known that a $(2k-1)$-spanner must
have $\tilde\Omega(n^{1+1/k})$ edges for some values of $k$, this is
conjectured to hold for all $k\in \N$, 
and a matching upper bound is obtained
by the Thorup-Zwick construction~\cite{TZ-05}. If an additive term
in the distance approximation is permitted, the multiplicative factor can be
brought arbitrarily close to $1$~\cite{EP04}. In contrast to routing and
distance approximation, there are extremely fast distributed algorithms
constructing sparse spanners.
Our long-range construction rests on an elegant
algorithm by Baswana and Sen~\cite{baswana07} that achieves 
stretch $2k-1$ vs.\ $\BO(n^{1+1/k})$ expected edges
within $\BO(k)$ rounds in the \CONGEST\ model. 


\section{Model}
\label{sec-model}
In this section we define the model of computation and formalize a few
concepts we use.

\subsection{The Computational Model} 
We follow the $\CONGEST(B)$ model as described in~\cite{Peleg:book}. The 
distributed system is represented by a simple, connected weighted
graph $G=(V,E,W)$, where $V$ is the set of nodes, $E$ is the set of
edges, and $W:E\to\N$ is the edge weight function.\footnote{We remark
that our results can be easily extended to non-negative edge weights by
employing appropriate symmetry breaking mechanisms.}
As a convention, we use $n$ to denote the number of nodes. We assume that all
edge weights are bounded by some polynomial in $n$, and that each node $v\in V$
has a unique identifier of $\BO(\log n)$ bits (we use $v$ to denote both the
node and its identifier).

Execution proceeds in global synchronous rounds, where
in each round, each node take the following three steps:
\begin{inparaenum}[(1)]
\item Receive messages sent by neighbors at the previous round,
\item perform local computation, and
\item send messages to neighbors.
\end{inparaenum}
Initially, nodes are aware only of their neighbors; input values (if any) are
assumed to be fed by the environment at time $0$. Output values are placed in
special output-registers. In each round, each edge can carry a message of $B$
bits for some given parameter $B$ of the model; we assume that $B\in \Theta(\log
n)$ throughout this paper.

A basic observation in this model is that we may assume, without loss
of generality, that we have a broadcast facility available, as
formalized in the following lemma.
\begin{lemma}
Suppose each $v\in V$ holds $m_v\ge0$ messages of $\BO(\log n)$ bits
each, for a total of $M\DEF \sum_{v\in V} m_v$ strings. Then all nodes in the
graph can receive these $M$ messages within $\BO(M+\HD)$ rounds.
\end{lemma}
\begin{proof}
Construct a BFS tree rooted at, say, the node $r$ with smallest identifier
($\BO(\HD)$ rounds). All nodes send their messages to their parents and
forward the messages received by their children to their parent as well, until
the root holds all messages. Since over no edge more than $M$ messages need to be
communicated, this requires $\BO(M+\HD)$ rounds. Finally all messages are
broadcast over the tree, completing in another $\BO(M+\HD)$ rounds.
\end{proof}
In the following, we will use this lemma implicitly whenever stating that some
information is ``broadcast'' or ``announced to all nodes.''

\subsection{General Concepts}
 We use  extensively ``soft'' asymptotic notation that ignores
 polylogarithmic factors. Formally, we say that $g(n)\in \tilde 
\BO(f(n))$ if and only if there exists a constant $c\in \R^+_0$ such that
$f(n)\leq g(n)\log^c(f(n))$ for all but finitely many values of $n\in \N$.
Anagolously, $f(n)\in \tilde\Omega(g(n))$ iff $g(n)\in
\tilde{\BO}(f(n))$, $\tilde{\Theta}(f(n))\DEF\tilde\BO(f(n))\cap
\tilde\Omega(f(n))$, $g(n)\in \tilde{o}(f(n))$ iff for each fixed
$c\in \R^+_0$ it holds that $\lim_{n\to \infty}g(n)\log^c(f(n))/f(n)=0$, and
$g(n)\in \tilde{\omega}(f(n))$ iff $f(n)\in \tilde{o}(g(n))$.

To model probabilistic computation, we assume that each node has
access to an infinite string of independent 
unbiased random bits.  When we say that a certain event occurs ``with high
probability'' (abbreviated ``w.h.p.''), we mean that the probability of the
event not occurring can be set to be less than $1/n^c$ for any desired constant
$c$, where the probability is taken over the strings of random bits.

\subsection{Some Graph-Theoretic Concepts} 
A \emph{path} $p$ connecting
$v,u\in V$ is a sequence of nodes
$\langle v=v_0,\ldots,v_k=u\rangle$ such that for
all $0\le i<k$, $(v_i,v_{i+1})$ is an edge in $G$. 
Let $\Paths(v,u)$ denote  the set of all paths
connecting nodes $v$ and $u$.  We use the following unweighted concepts.
\begin{compactitem}
\item The \emph{hop-length} of a path $p$, denoted
$\ell(p)$,  is the number of
edges in it.
\item The \emph{hop distance} $\Hd:V\times V\to \N_0$ is defined as
$\Hd(v,u):=\min\{\ell(p)\mid{p\in\Paths(v,u)}\}$.
\item The \emph{hop diameter} of a graph $G=(V,E,W)$ is
$\HD\DEF\max_{v,u\in V}\{\Hd(v,u)\}$.
\end{compactitem}
We use the following weighted concepts.
\begin{compactitem}
\item The \emph{weight} of a path $p$, denoted $W(p)$, is its total
edge weight, i.e., $W(p)\DEF\sum_{i=1}^{\ell(p)} W(v_{i-1},v_i)$.
\item The \emph{weighted distance} $\Wd:V\times V\to \R^+_0$
is defined by $\Wd(v,u)\DEF\min\{W(p)\mid{p\in\Paths(v,u)}\}$.
\item The \emph{weighted diameter} of $G$ is
$\WD\DEF\max\{\Wd(v,u)\mid{v,u\in V}\}$.
\end{compactitem}
The following concepts mix weighted and unweighted ones.
\begin{compactitem}
\item Given $h\in\N$ and two nodes $v,u\in V$ with hop
  distance $\Hd(v,u)\leq h$, we define the \emph{$h$-weighted
    distance} $\Wd_h(u,v)$ to be the weight of the lightest path
  connecting $v$ and $u$ with at most $h$ hops, i.e.,
  $\Wd_h(v,w)\DEF\min\{W(p)\mid p\in \Paths(v,w) \textrm{ and }
  \ell(p)\leq h\}$. If $\Hd(v,u)>h$, we define $\Wd_h(v,u)\DEF\infty$.
  (Note that $\Wd_h$  does not satisfy the
  triangle inequality.)
\item The \emph{shortest paths diameter} of a graph, denoted $\SPD$,
  is the maximal number of 
    hops in shortest paths:
    $\SPD\DEF\max_{u,v\in V}\Set{\min\Set{\ell(p)\mid W(p)=\Wd(u,v)}}$.
\end{compactitem}

Finally, given a node $v$ and an integer $i\ge0$, we define $\Ball_v(i)$ to be
the set  of the $i$ nodes that are closest to $v$ (according to $\Wd$, where
identifiers are used to break symmetry): $\displaystyle
\Ball_v(i)\DEF\Set{u\,:\,|\Set{w:(\Wd(v,u),u)\le(\Wd(v,w),w)}|\leq i}$. Note
that our concept of ball differs from the usual one: we define a ball by its
center and \emph{volume}, namely the number of nodes it contains (and not by its
center and radius).

We have the following immediate property.
\begin{lemma}
\label{lem:h}
Let $v,u\in V$. If $u\in\Ball_v(i)$ for some $i\in\N$
then $\Wd(v,u)=\Wd_j(v,u)$ for all
$j\ge i-1$.
\end{lemma}
\begin{proof}
Clearly $\Wd(v,u)\le\Wd_j(v,u)\le\Wd_{i-1}(v,u)$, and it therefore suffices to
show that $\Wd_{i-1}(v,u)=\Wd(v,u)$. Let $p=\Seq{v=v_0,v_1,\ldots,v_k=u}$ be a
shortest path from $v$ to $u$. Since edge weights are strictly positive, we have
that all the $k$ nodes $v_0,\ldots,v_{k-1}$ are strictly closer than $u$ to $v$.
Hence, since $u\in\Ball_v(i)$, we have that $i\ge k+1$. It follows that
$\Wd(v,u)=\Wd_{i-1}(v,u)$ and we are done.
\end{proof}

\section{Problem Statement and Lower Bounds}
\label{sec:prel}
\subsection{The Routing Problem}
\label{sec:problem}
In the \emph{routing table construction} problem (abbreviated $\rtc$), the local
input at a node is the weight of incident edges, and the output at each node $v$
consists of (i) a unique \emph{label} $\lambda(v)$ and (ii) a function
``$\Next_v$'' that takes a destination label $\lambda$ and produces a neighbor
of $v$, such that given the label $\lambda(u)$ of any node $u$, and starting
from any node $v$, we can reach $u$ from $v$ by following the $\Next$ pointers.
Formally, the requirement is as follows. Given a start node $v$ and a
destination label $\lambda(u)$, let $v_0=v$ and define
$v_{i+1}=\Next_{v_i}(\lambda(u))$ for $i\ge 0$. Then for some $i$ we must have
$v_i=u$.

The performance of a solution is measured in terms of its \emph{stretch}: A
route is said to have stretch $\rho\ge1$ if its total weight is no more than
$\rho$ times the weighted distance between its endpoints, and a solution to
$\rtc$ is said to have stretch $\rho$ if all the routes it induces have stretch
at most $\rho$.

\noindent\textbf{Variants.} Routing appears in many incarnations. We list a few
important variants below.

\emph{Name-independent routing.} Our definition of $\rtc$
allows for node relabeling. This is the case, as
mentioned above, in the Internet. The case where no such relabeling is allowed
(which can be formalized by requiring $\lambda$ to be the
identity function), is called \emph{name-independent} routing.

It can be shown that assigning new labels to the nodes is
unavoidable by proving that any (randomized) algorithm achieving polylogarithmic
(expected) stretch without relabeling must run for $\tilde\Omega(n)$ rounds.
Formally, we can prove the following.
\begin{theorem}
In the \CONGEST\ model, any algorithm for $\rtc$ that produces
name-independent stateful routing with expected average stretch 
$\rho$ requires $\Omega(n/(\rho^2\log n))$ time.
\end{theorem}

\emph{Stateful routing.} The routing problem as defined above is
\emph{stateless} in the sense that routing a packet is done regardless of the
path it traversed so far.  One may also consider \emph{stateful} routing, where
while being routed, a packet may gather information that helps it navigate later
(one embodiment of this idea in the Internet routing today is MPLS, where
packets are temporarily piggybacked with extra headers). Note that the set of
routes to a single destination in stateless routing must constitute a tree,
whereas in stateful routing even a single route may contain a cycle. Formally,
in stateful routing the label of the destination may change from one node to
another: The $\Next_v$ function outputs both the next hop (a neighbor node), and
a new label $\lambda_v$ used in the next hop.

\emph{Name-independent routing.} Our definition of $\rtc$ allows for node
relabeling. This is the case, as mentioned above, in the Internet. The case
where no such relabeling is allowed (which can be formalized by requiring
$\lambda$ to be the identity function), is called \emph{name-independent}
routing.

It can be shown that assigning new labels to the nodes is unavoidable by proving
that any (randomized) algorithm achieving polylogarithmic (expected) stretch
without relabeling must run for $\tilde\Omega(n)$ rounds. What might come as a
surprise here is that the result also applies to stateful routing.
\begin{theorem}
In the \CONGEST\ model, any algorithm for $\rtc$ that produces
name-independent routing with (expected) average stretch $\rho$ requires
$\Omega(n/(\rho^2\log n))$ time.
\end{theorem}

\subsection{The Distance Approximation Problem}

The \emph{distance approximation} problem is akin to the routing problem. Again,
each node $v$ outputs a label $\lambda(v)$, but now, $v$ needs to construct a
function $\mathrm{dist}_v: \lambda(V)\to \R^+$ (the table) such that for all
$w\in V$ it holds that $\mathrm{dist}_v(w)\geq \Wd(v,w)$. The stretch of the
approximation for a given node $w$ is $\mathrm{dist}_v(w)/\Wd(v,w)$, and the
solution has stretch $\rho\geq 1$, if $\mathrm{dist}_v(w)\leq \rho \Wd(v,w)$ for
all $v,w\in V$.

Similarly to routing, we call a scheme name-independent if $\lambda$
is the identity function. Since we require distances estimates to be
produced without communication, there is no ``stateful'' distance
approximation.

\subsection{Hardness of Name-Independent Distributed Table Construction}
\label{sec:naming}

While name-independence may be desirable, our routing and distance approximation
algorithm makes heavy use of relabeling. This is unavoidable for fast
construction, because, as the following two theorems show, any name-independent
scheme of polylogarithmic stretch requires $\tilde{\Omega}(n)$ rounds for table
construction. The lower bound holds even for stateful routing 
and average stretch. Moreover, since the construction below is generic,
intuitively it implies that there is no reasonable restriction, be it
in terms of topology, edge weights, or node degrees, that permits fast
construction of name-independent routing tables.%
\footnote{The lower bound graph can be adapted to be a balanced binary
  tree, weakening the lower bound on the stretch by factor $\log n$.}

\begin{theorem}\label{thm:lower_route_independent}
In the \CONGEST\ model, any name-independent routing scheme of (expected)
average stretch $\rho$ requires $\Omega(n/(\rho^2\log n))$ rounds for table
construction. This holds even if all edge weights are $1$, the graph is a tree
of constant depth, and the node identifiers are $1,\ldots,n$.
\end{theorem}
\begin{proof}
We assume w.l.o.g.\ that all set sizes we use in this proof are integer and that
nodes may send no more than exactly $\log n \in \N$ bits over each edge in each
round. Consider the following family of trees of depth $2$. The root is
connected to $n_1\in \Theta(\rho)$ inner nodes, each of which has $n_2$
children; denote by $I$ and $L$ the respective sets of nodes. All edges have
weight $1$, i.e., the maximal simple path weight is $4$.

We assign the identifiers $1,\ldots,n$ uniformly at random to the $n_1n_2$
leaves (w.l.o.g., we neglect that the total number of nodes is $n_1n_2+n_1+1$ in
the following and use $n$ instead). Consider any deterministic algorithm
constructing routing tables within $r\in \N$ rounds. From each node in $I$, the
root receives at most $r\log n$ bits, hence there are at most $2^{rn_1\log
n}$ possible routing tables at the root. Now consider the $n!/(n_2!)^{n_1}$
possible partitions of the leaf identifiers to the subtrees rooted at nodes from
$I$. We bound the number of such partitions for which a fixed routing table at
the root may serve a uniformly random routing request with probability at least $p$
correctly. This requirement translates to at least $pn$ identifiers being
exactly in the subtree where the routing table points to; we have
$\binom{n}{pn}$ possible choices for these identifiers. The remaining $(1-p)n$
identifiers may be distributed arbitrarily to the remaining subtrees. Depending
on the distribution of the $pn$ identifiers we already selected, the number of
possibilities for this may vary. Using standard arguments 
it can be shown that this quantity is maximized if the $pn$ identifiers are
distributed evenly among the subtrees, i.e., each of them contains $pn_2$ of
them. We conclude that no routing table can serve a uniform request with
probability at least $p$ for more than
$\binom{n}{pn}((1-p)n)!/((1-p)n_2)!^{n_1}$ of the possible input partitions.
Considering the number possible routing tables and the total number of input
partitions $n!/(n_2!)^{n_1}$, we have that
\begin{equation*}
\frac{(pn)!((1-p)n_2)!^{n_1}}{n_2!^{n_1}}=
\frac{n!/n_2!^{n_1}}{\binom{n}{pn}((1-p)n)!/((1-p)n_2)!^{n_1}}\leq
 2^{rn_1\log n}.
\end{equation*}

We distinguish two cases, the first being $p< e^2/n_1$. We seek to upper
bound $p$ in the second case as well, where $p\geq e^2/n_1$. Clearly the l.h.s.\
of the above inequality is increasing in $p\in [0,1]$. Together with Stirling's
approximation $x!\in e^{(1-o(1))x(\ln x-1)}$ we can bound
\begin{eqnarray*}
\frac{(pn)!((1-p)n_2!)^{n_1}}{n_2!^{n_1}}&\geq &
\frac{(e^2n_2)!((1-e^2/n_1)n_2)^{n_1}}{n_2!^{n_1}} \\
&\in & 
e^{(1-o(1))(e^2n_2(\ln (e^2n_2)-1)+n(1-e^2/n_1)(\ln((1-e^2/n_1)n_2)-1)
-n(\ln n_2-1))}\\
&\subseteq & 
e^{(1-o(1))(e^2n_2(\ln n_2+1)-e^2n_2(\ln n_2-1)+n \ln(1-e^2/n_1))}\\
&\subseteq & 
e^{(1-o(1))(2e^2n_2-e^2n_2)}\\
&=& e^{(e^2-o(1))n_2}.
\end{eqnarray*}
The assumption that $p\geq e^2/n_1$ thus implies (for sufficiently large $n$)
that $r>n_2/(n_1\log n)$.

Now condition on the event that for the given routing request the table
does not lead to the correct subtree. We fix the (uniformly random) subset of
leaf identifiers in the subtree $S$ the root's routing table points to, and
conclude that the set of remaining identifiers is a uniformly random subset of
$n-n_2-1$ leaf identifiers plus the destination's identifier. Moreover, the
destination is uniformly random from this subset and the remaining identifiers
are uniformly distributed among the remaining subtrees. We delete $S$ from the
graph (since clearly there is no reason to route to $S$ again) and examine the
next routing decision of the root. We observe that the situation is identical to
the initial setting except that $n_1$ is replaced by $n_1-1$. Note also that $S$
contained no valuable information: We deleted $S$ and the identifiers in $S$
from the graph, and any other information known to nodes in $S$ must have been
communicated to $S$ by the root. Hence, repeating the above arguments, we see
that the probability to find the destination in the second attempt conditioned
on the first having failed is at most $e^2/(n_1-1)$ or $r>n_2/(n_1\log n)$.
By induction on the number of routing attempts, we infer that for $i\in
\{1,\ldots,n_1/2\}$, the probability $p_i$ to succeed in the $i^{th}$ attempt to
route from the root node to the subtree containing the destination (conditional
on the previous attempts having failed) is upper bounded by $2e^2/n_1$ unless
$r>n_2/(n_1\log n)$.

Overall, the probability that a deterministic algorithm constructing routing
tables within $r\leq n_2/(n_1\log n)$ rounds fails to serve a uniformly random
routing request at the root for uniformly distributed leaf identifiers using
fewer than $n_1/2$ attempts (i.e., visits of the root on the routing path) is
lower bounded by
\begin{equation*}
\left(1-\frac{2e^2}{n_1}\right)^{n_1/2}\in \Omega(1).
\end{equation*}

Note that an analogous argument holds for routing requests issued at other
nodes, since they have a large probability to require routing to a different
subtree. Therefore, the average stretch of any deterministic routing
algorithm running for fewer than $n_2/(n_1\log n)$ rounds is at least
$\Omega(n_1)$. By Yao's principle, the expected average stretch of randomized
algorithms running for fewer than $n_2/(n_1\log n)$ rounds thus must also be
in $\Omega(n_1)$. Recalling that $n_1\in \Theta(\rho)$ and $n_2 = n/\rho$, we
get that $r\in o(n/(\rho^2\log n))$ rounds are insufficient to achieve
(expected) average stretch $\rho$, proving the statement of the theorem.
\end{proof}

A streamlined version of  the argument shows that a similar lower
bound applies to distance approximation.

\begin{theorem}\label{thm:lower_dist_independent}
In the \CONGEST\ model, any name-independent distance approximation scheme of
(expected) average stretch $\rho$ requires $\Omega(n/\log n)$ rounds for table
construction in graphs with edge weights of $1$ and $\omega_{\max}\in \BO(\rho)$
only. This holds even if the graph is a star and the node identifiers are
$1,\ldots,n$.
\end{theorem}
\begin{proof}
Again, we assume w.l.o.g.\ that all considered values are integer and that link
capacity is $\log n$ bits per round. Suppose $G$ is a star with $n$ leafs (we
neglect w.l.o.g.\ the center in the node count). All edges have weight
$\omega_{\max}$ with independent probability $1/2$; the remaining edges have
weight $1$.

Condition on the event that some fixed node's $v$ incident edge has weight $1$.
Thus, there are two possible path weights to other nodes: $\omega_{\max}+1$ and
$2$. Within $r$ rounds, the node receives at most $r\log n$ bits, yielding
$2^{r\log n}$ possible distance estimate configurations. In order to be
$\rho$-approximate for $\rho<(\omega_{\max}+1)/2$ and a given other leaf, $v$'s
table must output an estimate of at most $2\rho<\omega_{\max}+1$ in case the
leaf's edge has weight $2$ and at least $\omega_{\max}+1$ if the leaf's edge has
weight $\omega_{\max}$. Thus any given table can be correct for a given leaf for
only one of the two possible choices of the leaf's edge's weight. There are
$2^n$ possible edge weight assignments. By the above
observation, a fixed table is $\rho$-approximate for a given destination with
probability $1/2$. By Chernoff's bound, this implies that the probability that a
fixed table is correct for a fraction of $3/4$ of the destinations is bounded by
$2^{-\Omega(n)}$. By the union bound, it follows that for the given uniformly
random edge weight assignment, the probability that the computed table is
correct for a fraction of $3/4$ of the destinations is upper bounded by
$2^{-\Omega(n)}2^{r\log n}$. This implies that $r\in \Omega(n/\log n)$ or the
average stretch of node $v$'s table must be $\Omega(\omega_{\max})$.

By symmetry, the same applies to all nodes incident to an edge of weight $1$. By
Chernoff's bound, w.h.p.\ at least one quarter of the nodes satisfies this
property, i.e., the probability mass of the events where fewer than $n/4$ edges
have weight $1$ is negligible. By linearity of expectation, it follows that any
deterministic algorithm running for $o(n/\log n)$ rounds exhibits average
stretch $\Omega(\omega_{\max})$, and by Yao's principle this extends to the
expected stretch of randomized algorithms.
\end{proof}

Consequently, in the remainder of the paper we shall consider
name-dependent schemes only.

\subsection{Hardness of Diameter Estimation}
\label{sec:lb_diam}

In \cite{FHW-12}, it is shown that approximating the hop-diameter of a
network within a factor smaller than 1.5 cannot be done in the \CONGEST\ model
in $\tilde o(\sqrt n)$ time. Here, we prove a hardness result for the weighted
diameter, formally stated as follows.
\begin{theorem}\label{thm-lb-diam}
For any $\omega_{\max}\geq \sqrt{n}$, there is a function $\alpha(n)\in
\Omega(\omega_{\max}/\sqrt{n})$ such that the following holds. In the family of
weighted graphs of hop-diameter $\HD\in \BO(\log n)$ and edge weights $1$ and
$\omega_{\max}$ only, an (expected) $\alpha(n)$-approximation of the weighted
diameter requires $\tilde\Omega(\sqrt n)$ communication rounds in the \CONGEST\
model.
\end{theorem}
\begin{figure}[t]
  \centering \includegraphics[width=12cm]{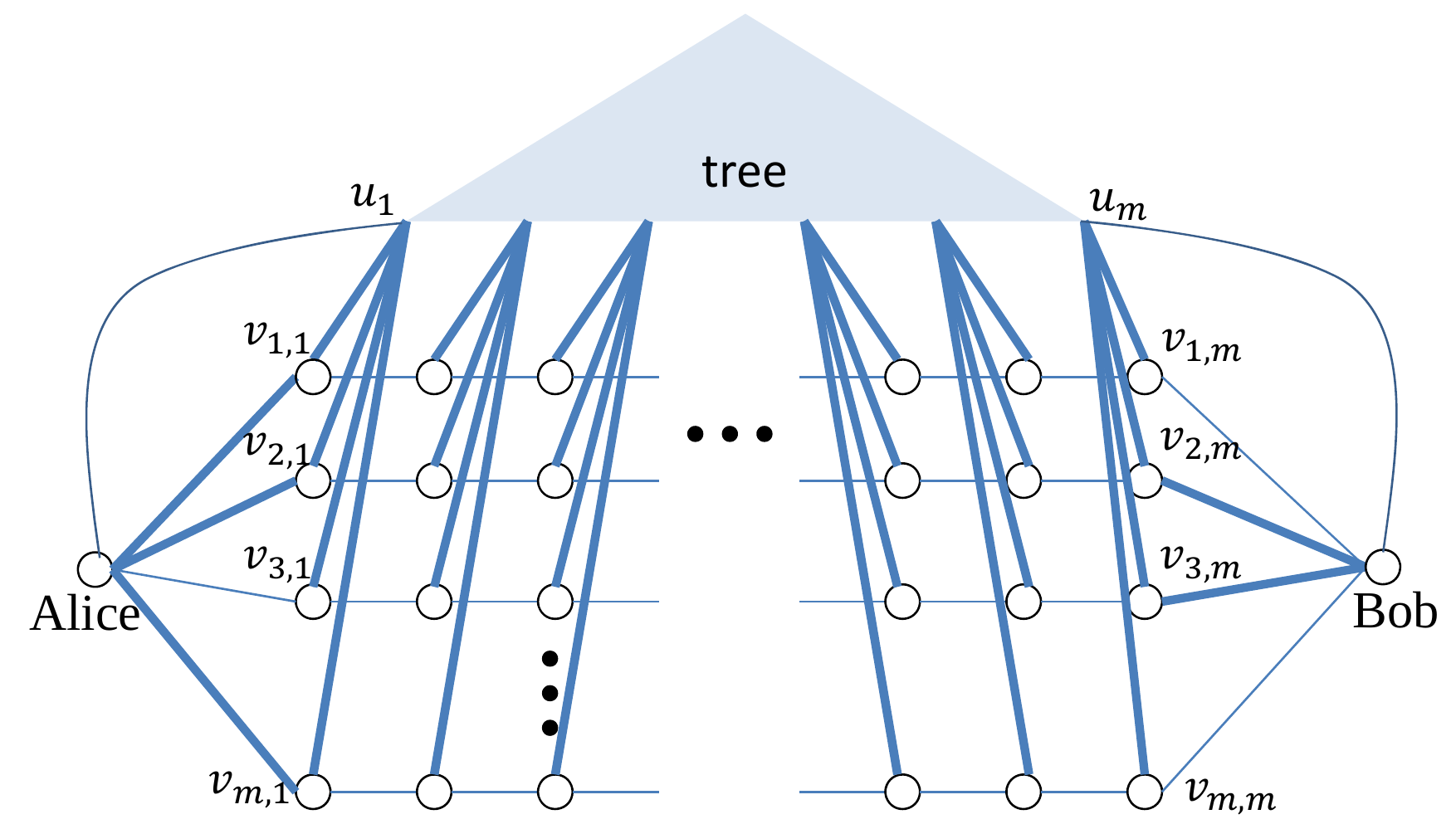} \caption{An
    illustration of  the graph
  used in the proof of
    \theoremref{thm-lb-diam}. Thick edges denote edges of weight
    $\omega_{\max}$, other edges are of weight $1$. The shaded
    triangle represents a binary tree}
  \label{fig-lb5}
\end{figure}
\begin{proof}[Proof sketch:]
  We construct a graph $G_n$ with $\Theta(n)$ nodes. Let $m=\sqrt n\in \N$. The
  graph consists of the following three conceptual parts. \figref{fig-lb5}
  illustrates a part of the construction.
  \begin{compactitem}
  \item Nodes $v_{i,j}$ for $1\le i,j\le m$. These nodes are connected
    as $m$ paths of length $m-1$. All path edges are of weight $1$.
  \item A star rooted at an \emph{Alice} node, where the children are
  $v_{1,1},\ldots,v_{m,1}$, and similarly, a star rooted at a \emph{Bob} node,
  whose leaves are $v_{m,1},\ldots,v_{m,m}$. We specify the weights of these
  edges later.
   \item 
     For each $1\le j\le m$ there is a node $u_j$ connected to all
     nodes $v_{i,j}$, $1\le i\le m$ in ``column'' $j$, with edges of
     weight $\omega_{\max}$. In addition, there is a binary tree whose
     leaves are the nodes $u_j$. All tree edges have weight $1$. Finally, \emph{Alice}
     and \emph{Bob} are connected to $u_1$ and $u_m$, respectively, by edges of weight
     $1$.
  \end{compactitem}
It is easy to see that the hop-diameter of $G_n$ is $\BO(\log n)$:
the hop-distance from any node to one of the nodes $u_j$ is $\BO(\log n)$,
and the distance between any two such nodes is also $\BO(\log n)$.
However, the majority of the short paths guaranteeing the small diameter passes
through very few nodes close to the root of the binary tree. Consequently, it
takes a long time to exchange a large number of bits between \emph{Alice} and \emph{Bob}, implying
that it is hard to decide set disjointness for sets held
by \emph{Alice} and \emph{Bob} in the \CONGEST\ model. Specifically, the following fact is a
direct corollary from \cite{DHKNPPW-11}.
\begin{fact}[\cite{DHKNPPW-11}]
\label{lb-fact}
Let $\cM\DEF\Set{1,\ldots,m}$. Suppose that node \emph{Alice}
holds a set $A\subseteq\cM$ and that node \emph{Bob} holds a set
$B\subseteq\cM$. Then finding whether $A\cap B=\emptyset$ takes
$\tilde\Omega(m)$ rounds in the \CONGEST\ model, even for randomized algorithms.  
\end{fact}

We now show that if the diameter of $G_n$ can be approximated within factor
$\omega_{\max}/\sqrt n$ in time $T$ in the \CONGEST\ model, then the set
disjointness problem problem can be solved in time $T+1$. To this end, we set
the edge weights of the stars rooted at \emph{Alice} and \emph{Bob} as follows: for all $i\in
\{1,\ldots,m\}$, the edge from \emph{Alice} to $v_{i,1}$ has weight $\omega_{\max}$ if
$i\in A$ and weight $1$ else; likewise, the edge from \emph{Bob} to $v_{i,m}$ has
weight $\omega_{\max}$ if $i\in B$ and weight $1$ else.

Note that given $A$ at \emph{Alice} and $B$ at \emph{Bob}, we can inform the nodes $v_{i,1}$
and $v_{i,m}$ of these weights in one round. Now run any algorithm that outputs
a value between $\WD$ and $\alpha(n)\WD\DEF \omega_{\max}\WD/(\sqrt{n}+C\log n)$
(for a suitable constant $C$) within $T$ rounds, and output ``$A$ and $B$ are
disjoint'' if the outcome is at most $\omega_{\max}$ and output ``$A$ and $B$
are not disjoint'' othwerwise.

It remains to show that the outcome of this computation is correct for any
inputs $A$ and $B$ and the statement of the theorem will follow from
\factref{lb-fact} (recall that the number of nodes of $G_n$ is $\Theta(n)$).
Suppose first that $A\cup B=\emptyset$. Then for each node $v_{i,j}$, there is a
path of at most $\sqrt{n}$ edges of weight $1$ connecting it to \emph{Alice} or
\emph{Bob}, and \emph{Alice} and \emph{Bob} are connected to all nodes in the
binary tree and each other via $\BO(\log n)$ hops in the binary tree (whose
edges have weight $1$ as well). Hence the weighted diameter of $G_n$ is
$\sqrt{n}+\BO(\log n)$ in this case and the output is correct (where we assume
that $C$ is sufficiently large to account for the $\BO(\log n)$ term). Now
suppose that $i\in A\cap B$. In this case each path from node $v_{i,1}$ to
\emph{Bob} contains an edge of weight $\omega_{\max}$, since the edges from
\emph{Alice} to $v_{i,1}$ and \emph{Bob} to $v_{i,m}$ as well as those
connecting $v_{i,j}$ to $u_j$ have weight $\omega_{\max}$. Hence, the weighted
distance from $v_{i,1}$ to \emph{Bob} is strictly larger than $\omega_{\max}$
and the output is correct as well. This shows that set disjointness is decided
correctly and therefore the proof is complete.
\end{proof}

\subsection{Hardness of Name-Dependent Distributed Table Construction}
\label{sec:lb}

A lower bound on name-dependent distance approximation follows directly from
\theoremref{thm-lb-diam}.
\begin{corollary}\label{coro-lb2}
For any $\omega_{\max}\geq \sqrt{n}$, there is a function $\alpha(n)\in
\Omega(\omega_{\max}/\sqrt{n})$ such that the following holds. In the family of
weighted graphs of hop-diameter $\HD\in \BO(\log n)$ and edge weights $1$
and $\omega_{\max}$ only, constructing labels of size $\tilde{o}(\sqrt{n})$ and
tables for distance approximation of (expected) stretch $\alpha(n)$ requires
$\tilde\Omega(\sqrt n)$ communication rounds in the \CONGEST\ model.
\end{corollary}
\begin{proof}
We use the same construction as in the previous proof, however, now we need to
solve the disjointness problem using the tables and lables. Using the same
setup, we run the assumed table and label construction algorithm. Afterwards, we
transmit, e.g., the label of \emph{Alice} to all nodes $v_{i,1}$. This takes
$\tilde{o}(\sqrt{n})$ rounds due to the size restriction of the labels. Then we
query the estimated distance to \emph{Alice} at the nodes $v_{i,1}$ and collect
the results at \emph{Alice}. Analogously to the proof of \theoremref{thm-lb-diam},
the maximum of these values is large if and only if the input satisfies that
$A\cap B=\emptyset$. Since transmitting the label costs only
$\tilde{o}(\sqrt{n})$ additional rounds, the same asymptotic lower bound as in
\theoremref{thm-lb-diam} follows.
\end{proof}

A variation of the theme shows that stateless routing requires
$\tilde{\Omega}(\sqrt{n})$ time.

\begin{corollary}\label{coro-lb}
For any $\omega_{\max}\geq \sqrt{n}$, there is a function $\alpha(n)\in
\Omega(\sqrt{\omega_{\max}}/n)$ such that the following holds. In the family of
weighted graphs of hop-diameter $\HD\in \BO(\log n)$ and edge weights $1$
and $\omega_{\max}$ only, constructing stateless routing
tables of (expected) stretch $\alpha(n)$ with labels of size
$\tilde{o}(\sqrt{n})$ requires $\tilde\Omega(\sqrt n)$ communication rounds in
the \CONGEST\ model.
\end{corollary}
\begin{proof}[Proof sketch:] We consider the same graph as in the proof of
\theoremref{thm-lb-diam} and input sets $A$ and $B$ at \emph{Alice} and
\emph{Bob}, respectively, but we use a different assignment of edge weights.
\begin{compactitem}
  \item All edges incident to a node in the binary tree have weight
  $\omega_{\max}$.
  \item For each $i\in \{1,\ldots,m\}$, the edge from \emph{Alice} to $v_{i,1}$
  has weight $\omega_{\max}$ if $i\in A$ and weight $1$ else. Likewise, the edge
  from \emph{Bob} to $v_{i,m}$ has weight $\omega_{\max}$ if $i\in B$ and
  otherwise weight $1$.
  \item The remaining edges (on the $m$ paths from $v_{i,1}$ to $v_{i,m}$) have
  weight $1$.
\end{compactitem}
Observe that the distance from \emph{Alice} to \emph{Bob} is $\sqrt{n}+1$ if
$A\cap B\neq \emptyset$ and strictly larger than $\omega_{\max}$ if $A\cap
B=\emptyset$. Once static routing tables for routing on paths of stretch at most
$\omega_{\max}/(\sqrt{n}+1)$ are set up, e.g.\ \emph{Bob} can decide whether $A$
and $B$ are disjoint as follows. \emph{Bob} sends its label to \emph{Alice} via
the binary tree (which takes time $\tilde{o}(\sqrt{n})$ if the label has size
$\tilde{o}(\sqrt{n})$). \emph{Alice} responds with ``$i$'' if the first
routing hop from \emph{Alice} to \emph{Bob} is node $v_{i,1}$ and $i\in
A$ (i.e., the weight of the edge is $1$), and ``$A\cap B=\emptyset$'' else (this
takes $\BO(\log n)$ rounds). \emph{Bob} then outputs ``$A\cap B\neq \emptyset$''
if \emph{Alice} responded with ``$i$'' and $i\in B$ (i.e., the weight of the
routing path is $\sqrt{n}+1$ since the edge from \emph{Bob} to $v_{i,m}$ has
weight $1$) and ``$A\cap B=\emptyset$'' otherwise.

If the output is ``$A\cap B\neq \emptyset$'', it is correct because $i\in A\cap
B$. On the other hand, if it is ``$A\cap B=\emptyset$'', the route from
\emph{Alice} to \emph{Bob} must contain an edge of weight $\omega_{\max}$,
implying by the stretch guarantee that there is no path of weight $\sqrt{n}+1$
from \emph{Alice} to \emph{Bob}. This in turn entails that $A\cap B=\emptyset$
due to the assignment of weights and we conclude that the output is correct also
in this case. Hence the statement of the corollary follows from
\factref{lb-fact}.
\end{proof}
We remark that \theoremref{thm:lower_dist_independent},
\theoremref{thm-lb-diam}, \corollaryref{coro-lb2}, and \corollaryref{coro-lb}
have in common that if edge weight $0$ is permitted, no stretch bound faster
than the stated lower bounds even if the only other feasible edge weight is $1$.

Finally, we note that the hop-diameter is also an obvious lower
bound on the time required to approximate the weighted diameter,
construct stateless routing tables, etc. since if the running time is smaller
than $\HD$, distant parts of the graph (in the sense of hop-distance) cannot
influence the local output.

\section{Routing Algorithm}
\label{sec:routing}

\textbf{Overview.} To construct routing tables, one needs to learn about paths.
Na\"\i ve distributed algorithms explore paths sequentially, adding one edge at
a time, leading to potentially linear complexity, since shortest weighted paths
may be very long in terms of the number of edges. Our basic idea is to break
hop-wise long paths into small pieces by means of random sampling. Specifically,
motivated by the $\tilde\Omega(\sqrt n)$ lower bound of
\theoremref{thm-lb-diam}, we select a random subset of
$\tilde{\Theta}{\sqrt{n}}$ nodes we call the routing \emph{skeleton}. It follows
that, w.h.p., (1)~any simple path of hop-length $\tilde\Omega(\sqrt{n})$
contains a skeleton node, and (2)~any node has a skeleton node among its closest
$\tilde\BO(\sqrt{n})$ nodes. The route that our scheme will select from a given
source to a given destination depends on their distance: If the destination is
one of the $\tilde\BO(\sqrt{n})$ nodes closest to the destination, routing will
be done using a ``short range scheme'' (see below); otherwise, the short range
scheme is used to route from the source to the nearest skeleton node, from
which, using another scheme we call ``long distance routing,'' we route to the
skeleton node closest to the destination node, and finally, another application
of the short range scheme brings us to the destination. Intuitively, we can
split the problem into the following tasks:
\begin{compactenum}
\item Short range scheme: how to route efficiently from each node to its
$\tilde{\Theta}(\sqrt{n})$ closest nodes including at least one
skeleton node, and, conversely, from a
skeleton node to all its ``subordinates'' (note the asymmetry in this case).
\item Skeleton routing scheme: how to route between skeleton nodes
efficiently.
\end{compactenum}

The short range scheme is described in \sectionref{sec:short}. We note that
since a straightforward application of multiple-source shortest paths may result
in linear time,  we develop a hierarchical structure to solve the short-range
routing. This hierarchy bears resemblance to the Thorup-Zwick distance oracle
algorithm \cite{TZ-05}. Our long distance routing is described in 
\sectionref{sec:skeleton}. The main challenge there is to build the skeleton
graph; since it might be too dense, we sparsify it ``on the fly'' while
constructing it. This construction is implemented by adapting the spanner
algorithm of Baswana and Sen \cite{baswana07} to our setting.

We start by describing the variant of the Bellman-Ford algorithm we use as a
basic building block in \sectionref{sec:bsp}.

\subsection{Bounded Shortest Paths}
\label{sec:bsp}

\begin{algorithm}[tb!]
\small
\SetKwInOut{Input}{input}
\SetKwInOut{Output}{computes}
\Input{
  $h$ \REM{range parameter: hop bound on path lengths, globally known}\\
  $\Delta$ \REM{overlap parameter: number of closest sources each node needs to
  detect, globally known}\\
  $\Src:V\to \mathrm{SID}\cup \{\bot\}$ \REM{each $v$
  knows $\Src(v)$; $\Src(v)=\bot$ means $v$ is not a
  source}\\
}
\Output{For all $t\in \{1,\ldots,h\}$: $t$-weighted distance and the next hop
from $v$ to each of the closest $\Delta$ source node sets using paths of
  at most $t$ edges (or all such sets, if there are at most $\Delta$ within $t$
  hops).
}
\lIf {$\Src(v)\neq \bot$\nllabel{bsp:init}}
  {$L_v(0):=\Set{(0,\Src(v),v)}$}
\lElse
  {$L_v(0):=\emptyset$ \nllabel{bsp:init2} \REM{initialization}}\\
\For {$t:=1$ \KwTo $h$\nllabel{bsp:iter}}{
  send $L_v(t-1)$ to all neighbors; $L_v(t):=\emptyset$\nllabel{bsp:repeat}\\
  \ForEach {neighbor $u$}{
    receive $L_u(t-1)$\\
    \ForEach(\REM{Bellman-Ford relaxation}) {$(d_u,s_u,\Next_u)\in L_u(t-1)$}{
      \If {$\exists(d_v,s_v,\Next_v)\in L_v(t)$ s.t.\ $s_v=s_u$}{
        \lIf {$(d_u\!+\!W(u,v),u)<(d_v,\Next_v)$} {
          \REM{comparisons are lexicographical}\\
          $L_v(t):= L_v(t)\setminus\Set{(d_v,s_v,\Next_u)}
          \cup\Set{(d_u\!+\!W(u,v),s_u,u)}$
        }
      }
      \lElse {$L_v(t):= L_v(t)\cup\Set{(d_u\!+\!W(u,v),s_u,u)}$}
    }
  }
  truncate $L_v(t)$ to smallest $\Delta$ entries 
  \REM{order is lexicographical}\nllabel{bsp:truncate}
}
\Return $(L_v(1),\ldots,L_v(t))$
\caption{$\BSP(h,\Delta,S)$: Bounded shortest paths, computed at node
$v\in V$.}\label{alg:bsp}
\end{algorithm}

We now describe a basic subroutine we use. Algorithm
$\BSP$, whose pseudo code is given in \algref{alg:bsp}, is
essentially a standard multiple-source distributed Bellman-Ford
algorithm, with two restrictions: first, the algorithm is run for
only $h$ rounds (cf.\ \lineref{bsp:iter}); and second, nodes never
report more than $\Delta$ sources (cf.\
\lineref{bsp:truncate}).

We consider a slightly extended variant of the algorithm: In the original
algorithm, each node is a ``source'' and the goal is to compute the distances of
all nodes to it. Here we assume that (i) not all nodes are sources, and (ii)
sets of nodes may act as a single source, as if there were 0-weight edges
connecting them. Both extensions are modeled by the $\textrm{source}$ function,
that maps a node to $\bot$ if it is not a source, or multiple nodes to the same
source ID if they are in the same source set. We use $S$ to denote the
set of \emph{sources}, i.e.,  $S=\Set{\Src(v)\mid v\in
  V}\setminus\Set\bot$, and for each $s\in S$,
the \emph{source nodes} of $s$ is $SN(s):=\Set{v\mid \Src(v)=s}$. Note that the
source function uniquely determines the source sets and vice versa. We assume
that a source ID can be encoded using $\BO(\log n)$ bits.

We analyze the algorithm leveraging the correctness of the basic Bellman-Ford
algorithm. To this end, let us define \algref{alg:bsp}* by omitting
\lineref{bsp:truncate} from \algref{alg:bsp} and fixing $h=n-1$. Observing that
\algref{alg:bsp}* is exactly the distributed Bellman-Ford algorithm, we may
conclude the following standard property.

\begin{lemma}\label{lem:bf}
Fix an execution of \algref{alg:bsp}*. Denote by $L_v^*(t)$ for some $0\le t\le
n-1$ and $v\in V$ the contents of the $L_v$ variable at node $v$ after $t$
iterations of \algref{alg:bsp}*. Then for each $(d,s,\Next)$ entry in $L_v^*(t)$
we have that $s$ is a source and $\Wd_t(v,s):=\min_{u\in
SN(s)}\{\Wd_t(v,u)\}=d$, namely $d$ is the length of the shortest path that
consists of at most $t$ edges from $v$ to any node $u$ in the source set of $s$.
Moreover, $\Next$ is the next node on that shortest path from $v$ to $u$.
\end{lemma}

\lemmaref{lem:bf} says that running only $h$ iterations is sufficient if we
are interested in paths of $h$ or less edges only. We now consider the effect of
repeatedly truncating the distance vector.

\begin{lemma}\label{lem:bsp_trunc}
Consider executions of \algref{alg:bsp} and of \algref{alg:bsp}* on the same
graph and with the same $\textrm{source}$ function. Let $L_v(t)$ and
$L_v^*(t)$ denote the contents of the $L_v$ variable at node $v$ after $t$
iterations under \algref{alg:bsp} and under \algref{alg:bsp}*, respectively.
Then $L_v(t)$ contains exactly the smallest $\Delta$ entries of $L_v^*(t)$ with
respect to lexicographical ordering (or the entire list, if $|L_v^*(t)|\leq
\Delta$).
\end{lemma}
\begin{proof}
By induction on $t$. The base case is $t=0$, and the lemma clearly holds upon
initialization (\lineref{bsp:init}). For the induction step, assume that the
lemma holds for $t-1\in \{0,\ldots,h-1\}$ at all nodes, and consider iteration
$t$ at some node $v$. By the induction hypothesis, we have that the message
received by $v$ from each neighbor $u$ at time $t$ under \algref{alg:bsp} is
exactly the top $\Delta$ entries sent by node $u$ at time $t$ under
\algref{alg:bsp}*, because these entries are computed at the end of iteration
$t-1$. The lemma therefore follows from the fact that for any $k\ge0$, the
smallest $k$ entries of a union of sets are contained in the union of the
smallest $k$ entries from each set.
\end{proof}

Note that the information provided by $L_v(h)$ is insufficient for routing:
since the $\Delta$ closest source node sets may differ between neighbors, it
may be the case that for some source identifier $s$ and two neighbors $v$ and
$u$ we have that $u$ is the next node from $v$ to $s$ in $L_v(h)$, but there is
no entry for source $s$ in $L_u(h)$! This occurs, for example, if in  iteration
$h$, $u$ learns about a source set closer than $SN(s)$, pushing $s$ out of
$L_u(h)$. However, since the algorithm returns $(L_v(1),\ldots,L_v(h))$ instead of
simply $L_v(h)$, we can still reconstruct the detected paths.
\begin{lemma}\label{lemma:bsp_route_stateful}
For any node $v$ and any entry $(d,s,\Next)\in L_v(h)$, a routing path of at
most $h$ hops from $v$ to a node in $s$ of weight $d$ can be constructed using
the $L$ tables at the nodes and a hop counter.
\end{lemma}
\begin{proof}
The routing decision for hop $t$ at the current node $v_{t-1}$ (where $v_0:=v$)
is made by looking up the entry $(d_{t-1},s,\Next)\in L_v(h-(t-1))$. We show by
induction on the length $\ell\leq h$ of a shortest path from $v$ to its closest
node $u\in SN(s)$ that such an entry always exists. Note that by
Lemmas~\ref{lem:bf} and~\ref{lem:bsp_trunc}, such an entry satisfies that
$d_{t-1}=\Wd_{h-(t-1)}(v_{t-1},u)$ and thus the constructed path has weight
$\Wd_h(v,u)=d$. Trivially, the claim is true for $\ell=0$ by initialization of
the lists $L_u(0)$, $u\in V$.

Now suppose the claim holds for $\ell\in \N_0$ and consider node $v$ with entry
$(\Wd(v,u),s,\Next)\in L_v(h)$. Suppose $w$ is the neighbor of $v$ which is next
on the shortest $\ell$-hop path from $v$ to $u$. Hence it is the endpoint of a
of a shortest $(\ell-1)$-hop path from $w$ to $u$, and there is no shorter path
from $w$ to any node in $SN(s)$ of at most $\ell-1$ hops (otherwise there would
be a shorter path of at most $\ell$ hops from $v$ to a node in $SN(s)$).
Therefore, by \lemmaref{lem:bf}, $(\Wd_{h-1}(w,u),s,\Next_w)\in L_u^*(h-1)$ for
some $\Next_w$. Assuming for contradiction that
$(\Wd_{h-1}(w,u),s,\Next_w)\notin L_u(h-1)$ implied, by
\lemmaref{lem:bsp_trunc}, that there are $\Delta$ entries $(d,s',\Next)\in
L_u(h-1)$ that are lexicographically smaller than $(\Wd_{h-1}(w,u),s,\Next_w)$.
Node $w$ would send these smaller entries in iteration $h$ of \algref{alg:bsp},
yielding the contradiction that $(\Wd(v,u),s,\Next_w)\notin L_v(h)$. It follows
that indeed $(\Wd_{h-1}(w,u),s,\Next_w)\in L_u(h-1)$ and the proof concludes.
\end{proof}

We summarize the properties of \algref{alg:bsp} with the following
theorem. 
\begin{theorem}\label{thm:bsp}
\algref{alg:bsp} computes the $h$-weighted distance and next hop of a shortest
path of at most $h$ edges from each node to its closest $\Delta$ source sets.
Each node on the corresponding shortest path can determine the next hop on the
path out of the number of preceding hops and the output of the algorithm.
The time complexity of \algref{alg:bsp} in the \CONGEST\ model is $\BO(\Delta
h)$ rounds.
\end{theorem}
\begin{proof}
Correctness follows from Lemmas \ref{lem:bf} and \ref{lem:bsp_trunc}.
\lemmaref{lemma:bsp_route_stateful} proves that the paths can be reconstructed
as stated. The time complexity follows from the fact that the algorithm runs for
$h$ iterations, and each iteration can be implemented  in $\BO(\Delta)$ rounds
in the \CONGEST\ model since the messages contain $\BO(\Delta)$ IDs and
distances.
\end{proof}

\emph{Stateless routing.} 
The routing mechanism suggested by \lemmaref{lemma:bsp_route_stateful} has the
disadvantage that it is stateful, as the routing decision depends on the number
of previous routing hops. It is easy to make it stateless: at each
node, a packet is directed toward the hop that reported the best
distance estimate, i.e., the next hop to take at node $v$ for
destination $s$ is $\arg\min_{\Next}\Set{d\,:\,(d,s,\Next)\in
  \bigcup_t L_v(t)}$.
\begin{corollary}\label{coro:bsp_route_stateless}
For any node $v$ and any entry $(d,s,\Next)\in L_v(h)$, a routing path of at
from $v$ to a node in $s$ of weight $d$ can be constructed using the local
knowledge of the nodes only.
\end{corollary}
\begin{proof}
\lemmaref{lemma:bsp_route_stateful} shows that if a node $w$ follows the
$\Next_w$ pointer of \emph{any} entry $(d_w,s,\Next_w)\in L_w(t)$ for \emph{any}
$t\in \{1,\ldots,h\}$, node $\Next_w$ has an entry
$(d'-W(w,\Next_w),s,\Next_{\Next_w})\in L_w(t-1)$. We thus can simply choose to
follow at each node $v$ the $\Next$ pointer of entry $(d,s,\Next)\in
\bigcup_{t\in \{1,\ldots,h\}}L_v(t)$ with minimal $d$ and are guaranteed to
eventually arrive at some node in $s$ using a path of weight at most $d$.
\end{proof}
Note that in general we cannot guarantee that the constructed path has at most
$h$ hops when applying this mechanism; this holds true, however,  if we
are routing to one of the $h$ nodes 
closest to the source of the
routing request (by \lemmaref{lem:h}). This observation will be crucial for
making our general routing scheme stateless.

\subsection{The Short-Range Scheme}
\label{sec:short}

With Algorithm $\BSP$ at hand, we can now describe our short-range routing
scheme. Our goal is to allow each node to find a route to each of its closest
$\tilde\Theta(\sqrt{n})$ neighbors. A na\"\i ve application of Algorithm $\BSP$,
where all nodes are sources, would set the overlap parameter to
$\tilde\Theta(\sqrt{n})$ (this is the number of nodes we want to know about),
and the range parameter to $\tilde\Theta(\sqrt{n})$ too (in order to find the
closest $\tilde\Theta(\sqrt{n})$ nodes it suffices to go to this hop-distance,
cf.\ \lemmaref{lem:h}). However, \theoremref{thm:bsp} tells us that in this
case, the time complexity would be $\BO(\Delta h)\subset\tilde\BO(n)$, a far cry
from the $\tilde\Omega(\sqrt n)$ lower bound from Corollaries \ref{coro-lb2} and
\ref{coro-lb}. Our solution is a hierarchical bootstrapping process that
converges in double-exponential speed. We show that the stretch is proportional
to the number of stages in the hierarchy.

\subsubsection*{The Construction}
The construction is done iteratively in $L$ stages. In the interest of clarity
we describe the construction intuitively first and then formalize it. The idea
is that on the one hand we want to spend
at most a certain amount of time, but on the other hand with each stage try to
reduce the number of landmarks as quickly as possible. This approach is
the spirit of Thorup-Zwick distance oracles and routing
schemes~\cite{TZ-routing,TZ-05}, and it is also used in a distributed
fashion in~\cite{DDP}. The difficulty lies in constructing such a
hierarchy quickly.%
\footnote{In \cite{DDP}, distance sketches are constructed
  distributedly using exhaustive search with
  respect to distances, i.e., Bellmann-Ford is run for
  sufficiently many iterations until all routes become stable. This
  approach has time complexity $\Omega(\SPD)$ and therefore cannot guarantee a
  running time of $\tilde{o}(n)$ on all graphs of diameter $\HD\in \tilde{o}(n)$.
}

\begin{wrapfigure}{r}{2.2in}
\centering
\includegraphics[width=1.8in]{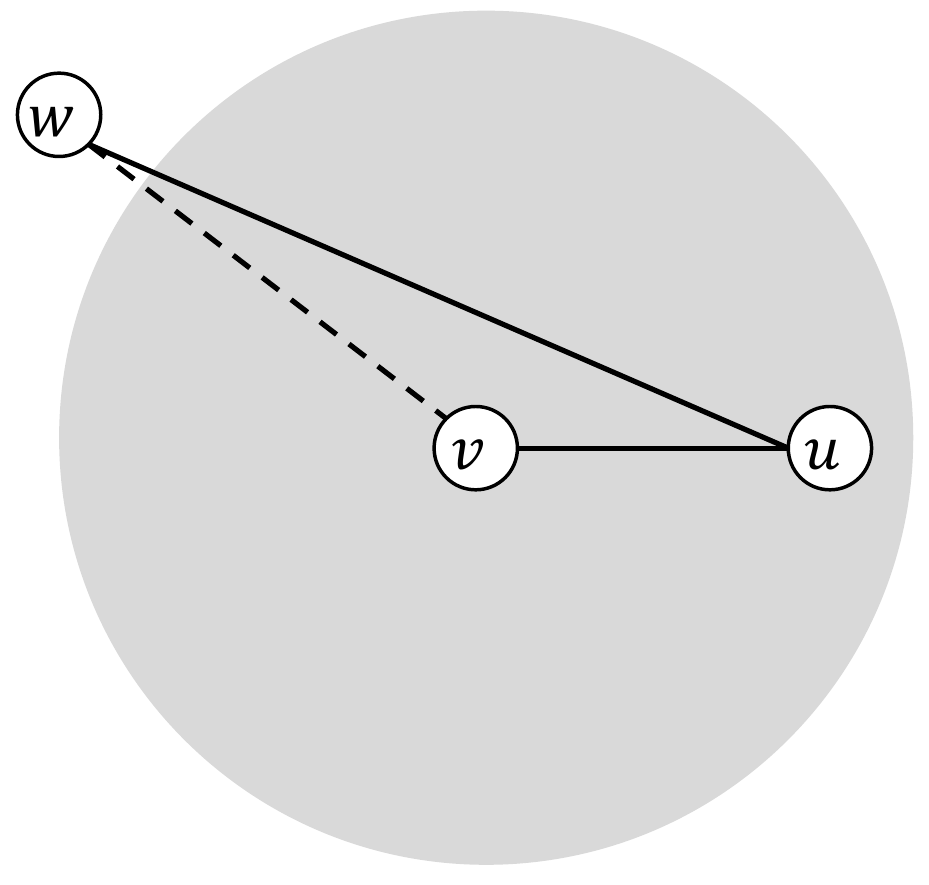}
\caption{\small The distance from $v$ to $w$ is at least one third
of the length of the route from $v$ to $w$ via $u$.}
\label{fig:tri}
\end{wrapfigure}
The sets of landmarks, denoted $S_1,\ldots,S_L$, are sampled uniformly
and independently 
at random 
without any coordination overhead, with $S_0\DEF V$, and 
$S_{i}\subseteq S_{i-1}$ for $1\le i\le L$.
In the $i^{th}$ stage, each node finds
a route to the closest node in $S_i$ as well as
to all nodes in $S_{i-1}$ that are closer to it.
This property 
allows us to bound the routing stretch. The basic argument is a simple
application of the triangle inequality (see \figref{fig:tri}): Consider a route  from node
$v$ to node $w$. If there is a node $u\in S_1$ that is closer to $v$ than
$w$, then the route of shortest paths via $u$ has stretch at most $3$.
It is therefore sufficient for $v$
to determine (the next hop of) least-weight routes to nodes in
$S_0\DEF V$ that are closer to it than the closest node in $S_1$
only. Using double induction, we can bound the stretch of the
multi-stage application of this technique we employ.

To this end, in each stage we invoke \algref{alg:bsp} with source set
$S_{i-1}$. We now explain how to choose the parameters $h_i$ and
$\Delta_i$ for this invocation. Let $p_i$ be the probability of 
a node to be selected into $S_i$. Then w.h.p., each node $v$ has a member of
$S_i$ among the $\BO(\log n/p_i)$ nodes closest to $v$. Hence, this is a good
choice for the distance parameter $h_i$. The expected number of nodes from
$S_{i-1}$ among the  $h_i$ nodes closest to a given node is $p_{i-1}h_i$.
Applying Chernoff's bound shows that this number is bounded by
$\BO(p_{i-1}h_i)=\BO(p_{i-1}\log n /p_i)$ w.h.p. This is an upper bound on
the number of sources that need to be detected by each node and therefore is our
choice of the overlap parameter $\Delta_i$.

The resulting running time of the call to Algorithm $\BSP$ is
$\BO(h_i\Delta_i)\subset \tilde{\BO}(p_{i-1}/p_i^2)$. Since this is the
dominating term in the running time in each stage, it is now easy to determine
the sampling probabilities: neglecting polylogarithmic factors, we get the
simple recursion $p_i=\sqrt{p_{i-1}/T}$, where $T$ is the desired running time
and $p_0\DEF 1$. For example, if we want to ensure a running time bound of
$T\in \tilde{\BO}(\sqrt{n})$, we obtain: 
\begin{compactitem}
\item sampling probabilities of $n^{-1/4}, n^{-3/8}, n^{-7/16},\ldots$, i.e.,
$p_i=n^{-(2^i-1)/2^{i+1}}$;
\item expected set sizes of
$\Theta(n^{3/4}),\Theta(n^{5/8}),\Theta(n^{9/16}),\ldots$, i.e., $|S_i|\in
\Theta(n^{1/2+1/2^i})$ w.h.p.;
\item range parameters of $\Theta(n^{1/4}\log n),\Theta(n^{3/8}\log
n),\Theta(n^{7/16}\log n)$, i.e., $h_i\in \Theta(n^{1/2-1/2^{i+1}}\log n)$;
\item overlap parameters of $\Theta(n^{1/4}\log n),\Theta(n^{1/8}\log
n),\Theta(n^{1/16}\log n)$, i.e., $\Delta_i\in \Theta(n^{1/2^{i+1}}\log n)$.
\end{compactitem}
(Note that $L=\log \log n$ stages suffice to ensure that $S_L\in
\Theta(\sqrt{n})$ w.h.p.) Running Algorithm $\BSP$ with parameters as above, we
get that w.h.p., after $\tilde\BO(\sqrt n)$ time, each node knows of the closest
$\Delta_i$ nodes from $S_{i-1}$ and how to route to them for
all $1\le i\le L$. But this is not sufficient: we also need to be able to route
back from the nodes in $S_i$.

Given a node $v$, define $\Lead_i(v)$ to be the node closest to $v$ in $S_i$
(symmetry broken by identifiers), and let $C_v(i)\DEF \{u\in
V\,|\,\Lead_i(u)=v\}$, i.e., for each stage $i$, the sets $C_v(i)$ are a Voronoi
decomposition of $V$ with centers $\Lead_i(V)$.
Note that routing from $\Lead_i(v)$ to $C_v(i)$ is not as simple as thee other
direction: While the depth of the tree rooted at $\Lead_i(v)$ is bounded by
$h_i$, there is no non-trivial upper bound on the number of nodes in the tree.
This can be solved by a number of standard techniques for tree routing (e.g.,
\cite{SK}). To minimize space consumption, we use the technique of
\cite{TZ-routing}, which constructs routing tables of size $\tilde{\BO}(1)$ and
node labels of $\BO(\log n)$ bits in $\BO(h_i)$ time.
In a nutshell, the idea is first to count the sizes of subtrees (which can be
done in $\BO(h_i)$ rounds) and then construct ``mini routing tables'' for the
``heavy'' part of the tree, where a node is considered heavy if its subtree
contains at least $n/\lceil\sqrt{\log n}\rceil$ nodes. Then this process
is applied recursively in the subtrees rooted at children of heavy
nodes. From the description in~\cite{TZ-routing}, one can verify that each
recursive step of the construction can be performed in time $\tilde\BO(h_i)$ in
a tree of depth $h_i$ in the $\CONGEST(\log n)$ model. There are at most
$\log_{\sqrt{\log n}}n$ recursive steps, summing up to a total of
$\tilde\BO(h_i)$ rounds to construct labels and routing tables.

Formally, given natural numbers $n$ and $L\leq \log \log n$, we define the
following for $1\le i\le L$.
\begin{compactitem}
\item  $p_0\DEF 1$, and $p_i\DEF(\sqrt{n})^{-(2^L/(2^L-1))(2^i-1)/2^i}$.
\item For each node $v$, $\Lead_v(i)$ is  the node from $S_i$ closest to
$v$ (ties broken by hop distance and ID).
\item   For each $u\in S_{i}$, define
$C_u(i)\DEF\Set{v\mid\Lead_v(i)=u}$, and $C_u(0)\DEF \{u\}$.
\item For each node $v$, define $H_v(i)\DEF \{u\in S_{i-1}\,|\,\Wd(v,u)\leq
\Wd(v,\Lead_v(i))\}$.
\end{compactitem}
Our construction maintains (w.h.p.)\ the following
properties at stage $i\in \{1,\ldots,L\}$.\vspace*{1ex}

\newlength{\fparwidth}\addtolength{\fparwidth}{\textwidth}\addtolength{\fparwidth}{-13pt}
\noindent\fbox{\begin{minipage}{\fparwidth}
\begin{compactenum}[(1)]
\item\label{prop-prob} $S_i$ is a uniformly random subset of $S_{i-1}$, 
where $\displaystyle\Pr[v\in S_i]=p_i$ and $\displaystyle\Pr[v\in
S_i\mid v\in S_{i-1}]= p_i/p_{i-1}=(\sqrt{n})^{-2^L/(2^i(2^L-1))}$.
\item \label{prop-y}
For any node $v$, it it is possible to route from $v$ to $\Lead_v(i)$
on a least-weight path.
\item \label{prop-f}
For any node $v$, it is possible to compute $\Lead_v(i)$ and
$\Wd(v,\Lead_v(i))$ from the label of $v$.
\item \label{prop-c}
For any node $u\in S_i$, it is possible to route from $u$ to any node $w\in
C_u(i)$ on a least-weight path. 
\item\label{prop-h}
For any node $v$,  $H_v(i)$ is locally known at $v$, and it is
possible to route
from $v$ to any node $u\in H_v(i)$ on a least-weight path (whose weight is known
at $v$).
\end{compactenum}
\end{minipage}
}
\medskip

Suppose that we have such a hierarchy of $L$ stages. Then, given the label of
any node $w\in \bigcup_{1\le i\le L}\bigcup_{u\in H_v(i)}C_u(i-1)$, node $v$ can
route a message to $w$ as follows: First, find some $i\in
\{1,\ldots,L\}$ such that $w\in C_u(i-1)$ for some $u\in H_v(i)$
(cf.\ \pprtyref{prop-f} and \pprtyref{prop-h} of the construction). The route
from $v$ to $w$ is then defined by the concatenation of two shortest paths:
the one from $v$ to $u$, and the one from $u$ to $w$ (cf.\ \pprtyref{prop-c} and
\pprtyref{prop-h}). Moreover, the long-range scheme will make sure that we can
always route to any destination via the closest skeleton nodes in $S_L$, which
is feasible due to \pprtyref{prop-y} and \pprtyref{prop-c}. By always choosing
from the available routes such that the weight of the computed route is minimal
(which can be done by \pprtyref{prop-f} and \pprtyref{prop-h} for the
short-range construction, and will also be possible for the long-range scheme),
routing becomes stateless.

\subsubsection*{Stretch Analysis}

We now bound the weight of the routes constructed by the stated scheme with
respect to the weight of the shortest paths. We note that the argument for the
general case is similar in spirit to the simple case of $i=1$ illustrated
in \figref{fig:tri}. We start with the following key lemma.
\begin{lemma}
\label{lem-sep}
Suppose that for $v,w\in V$ and $1\le j\le L$ we have that
$w\notin\bigcup_{i=1}^j\bigcup_{u\in H_v(i)}C_{u}(i-1)$. Then (a)
$\Wd(v,\Lead_v(j))\leq (2{j}-1)\Wd(v,w)$, and (b) $\Wd(w,\Lead_w(j))\leq 
2j\Wd(v,w)$.
\end{lemma}
\begin{proof}
We prove the lemma by induction on $i$, for a fixed $j$.
More specifically, we show for each $1\leq i\leq j$ that (a) $\Wd(v,v_{i})\leq
(2i-1)\Wd(v,w)$ and (b) $\Wd(w,w_{i})\leq  2i\Wd(v,w)$. For the basis of the
induction, consider $i=0$ in Statement (b). In this case, since $S_0=V$, we
have that, $\Lead_w(0)=w$ and Statement (b) holds because $\Wd(v,w)\ge
0=\Wd(w,w)$.

For the inductive step, assume that Statement (b) holds for $0\leq i<j$ and
consider $i+1$. Since trivially $w\in C_{\Lead_w(i)}(i-1)$, the premise of the
lemma implies that $\Lead_w(i)\not \in H_v(i+1)$. However, $\Lead_v(i+1)\in
H_v(i+1)$, and hence we obtain
\begin{eqntext}
\Wd(v,\Lead_v(i+1))&\leq&\Wd(v,\Lead_w(i))\\
&\leq& \Wd(v,w)+\Wd(w,\Lead_w(i))& \text{triangle inequality}\\
&\leq& (2i+1)\Wd(v,w)& \text{by induction hypothesis}
\end{eqntext}
This proves part (a) of the claim. Using the above inequality we also
obtain
\begin{eqntext}
\Wd(w,\Lead_w(i+1))
&\leq&\Wd(w,\Lead_v(i+1))& \text{$\Wd(w,\Lead_w(i+1))\le\Wd(w,u)$
for $u\in S_{i+1}$}\\
&\leq &\Wd(w,v)+\Wd(v,\Lead_v(i+1))& \text{triangle inequality}\\
&\leq &(2i+2)\Wd(v,w)& \text{by the proof of part (a)},
\end{eqntext}
which proves part (b) of the claim, completing the inductive step.
\end{proof}
\lemmaref{lem-sep} allows us to prove the following positive result.
\begin{corollary}
\label{cor-short}
Let $v,w\in V$, and let $1\le i_0\le L$ be minimal such that
$\Lead_w(i_0-1)\in H_v(i_0)$.
Then $\Wd(v,\Lead_w(i_0-1))+\Wd(\Lead_w(i_0-1),w)\le(4i_0-3)\Wd(v,w)\in
\BO(L\cdot\Wd(v,w))$.
\end{corollary}
\begin{proof}
Note that
\begin{eqntext}
\Wd(v,\Lead_w(i_0-1))+\Wd(\Lead_w(i_0-1),w)
&\le& \Wd(v,w)+2\Wd(w,\Lead_w(i_0-1)) & \text{triangle inequality}\\
&\le& \Wd(v,w)+4(i_0-1)\Wd(v,w)& \text{\lemmaref{lem-sep}}\\
&=& (4i_0-3)\Wd(v,w)
\end{eqntext}
and the corollary is proved.
\end{proof}
On the other hand, if there is no $i_0$ as in the corollary, we can conclude
from \lemmaref{lem-sep} that routing via the skeleton nodes closest to source
and destination, respectively, incurs bounded stretch.

\subsubsection*{Implementation and Time Complexity}
We now explain how to construct the hierarchy efficiently in more detail, and
analyze the time complexity of the construction. \algref{algo:close} gives the
pseudocode of the above scheme. The algorithm is parametrized by the total
number of nodes $n$ and the number $L$ of hierarchy stages. Appropriate
constants $c$ and $c'$ are supposed to be predefined in accordance with the required lower bound
on the probability of success.\footnote{%
  One can verify the properties of the construction and restart a
  failed iteration within $\BO(\HD)$ time if desired, implying that
  the stretch guarantee becomes deterministic and the running time
  probabilistically bounded instead.
}

\begin{algorithm}[ht!]
\small
\SetKwInOut{Input}{input}
\SetKwInOut{Output}{computes}
\Input{$n\in \N$ \REM{number of nodes}\\
$L\in \{1,\ldots,\log \log n\}$ \REM{number of stages in the hierarchy}}
\Output{$l_v\in \{0,\ldots,L\}$ \REM{level of $v$; $v\in S_i \Leftrightarrow
  l_v\geq i$}\\
$\forall i\in \{1,\ldots,L\}: \Lead_v(i)\in S_i$\REM{closest node in $S_i$}\\
$\forall i\in \{1,\ldots,L\}: H_v(i)=\{w\in S_{i-1}\,|\,\Wd(v,w)\leq
\Wd(v,\Lead_i(v)\}$\\
$\forall i\in \{1,\ldots,L\}\,\forall u\in H_v(i): \Next_v(u),d_v(u)$ \REM{next
routing hop ($v$ if $v=u$) and distance to $u$}
}
\lFor{$i\in \{0,\ldots,L\}$}{$p_i:=(\sqrt{n})^{-(2^L/(2^L-1))(2^i-1)/2^i}$\\}
$l_v:= i \text{ with probability }p_i-\sum_{j=i+1}^{L}p_j \text{ for }i\in
\{0,\ldots,L\}$\nllabel{line:set}\\
\For{$i\in \{1,\ldots,L\}$}{
  $h_i:=c\cdot \log n /p_i$ \REM{$c$ and $c'$ are predefined constants
  controlling the probability of failure}\\
  $\Delta_i:=c'\cdot h_i p_{i-1}$\\
  \lIf{$l_v\ge i$}{$\Src(v):=(v,l_v)$} \lElse{$\Src(v):=\bot$}\\
  $(L_v(1),\ldots,L_v(h_i)):=\BSP(h_i,\Delta_i,\Src)$ \REM{only $L_v(h_i)$
  needed}\\
  $H_v(i):=\emptyset$\\
  \Repeat{$l_u\geq i$}{
    let $(d,(u,l_u),w)$ be the next entry in $L_v(h_i)$ in ascending
    lexicographic order\\
    $H_v(i):=H_v(i)\cup \{u\}$\\
    $\Next_v(u):=w$; $d_v(u):=d$ \REM{exact shortest paths, no distinction of
    stages needed}
  }
  $\Lead_v(i):=u$
  \REM{$u$ is the node from $S_{i+1}$ closest to $v$}\\
  construct labels of stage $i$
}
\caption{Distributed construction of data structure for close-distance
  routing at $v\in V$.
}
\label{algo:close}
\end{algorithm}

Choosing the sets $S_i$ is performed locally without communication. Each node
$v$ has level $l_v$ chosen independently so that $\Pr[l_v\geq i]=p_i$
(\lineref{line:set}). Setting $S_i\DEF \{v\in V\,|\,l_v\geq i\}$ as indicated in
the algorithm thus satisfies \pprtyref{prop-prob}. In addition, the following
properties are easily derived using the Chernoff bound, and we state them
without proof.

\begin{lemma}\label{lem-short-corr1-general}
For appropriate choices of the constants $c$, $c'$ in
\algref{algo:close}, for all $1\le i\le L$ it holds w.h.p.\ that:
\begin{compactitem}
\item $|S_i|\in \Theta(p_i n)$ $(|S_0|=n)$.
\item For all $v\in V$, $|S_i\cap \Ball_v(h_i)|\in \Theta(\log n)$.
\item For all $v\in V$, $H_v(i)\subset \Ball_v(h_i)$.
\item For all $v\in V$, $|H_v(i)|\in \Theta(h_i p_{i-1})=\Theta(p_{i-1}\log
n / p_i)$.
\item For all $v\in V$, $\Delta_i\geq |H_v(i)|$.
\end{compactitem}
\end{lemma}

By these properties and \theoremref{thm:bsp}, \pprtyref{prop-h} is satisfied
w.h.p., because we invoke Algorithm $\BSP$ with sources $S_i$, depth parameter
$h_i$, and overlap parameter $\Delta_i$: after this invocation, each node $v$
can identify the set $H_v(i)$ and route to any $u\in H_v(i)$ on a path of known
weight; since $H_v(i)\subset \Ball_v(h_i)$ w.h.p., these routing paths are
shortest paths. Moreover, the invocation of $\BSP(h,\Delta,S_{i})$ allows each
node $v$ also to learn what is $\Lead_v(i)$ and route to it on a shortest path
of known weight, establishing \pprtyref{prop-y}. In order to satisfy
\pprtyref{prop-f}, we simply add $\Lead_v(i)$ and
$\Wd(v,\Lead_v(i))=d_v(\Lead_v(i))$ to the label of $v$ for all $i$. As
discussed earlier, routing tables of size $\log^{\BO(1)}n$ and labels of
size $(1+o(1))\log n$ to route within $C_{\Lead_v(i)}$ can be constructed within
$\tilde\BO(h_i)$ rounds using the scheme from~\cite{TZ-routing}, and we add the
respective tree label to $v$'s label to ensure \pprtyref{prop-c}.

We can therefore summarize the complexity of the construction as follows.
\begin{lemma}
\label{lem-short-perf}
Given $1\le L\le\log\log n$, constructing the $L$-stages short-range routing
tables and labels can be done in $\BO(L(\sqrt{n})^{2^L/(2^L-1)}\log^2 n)\subset
\tilde{\BO}((\sqrt{n})^{2^L/(2^L-1)})$ rounds, and the total label size of a
node is $\BO(L\log n)$.
\end{lemma}
\begin{proof}
The implementation of stage $i\in \{1,\ldots,L\}$ involves invoking $\BSP$ with
parameters $h_i$ and $\Delta_i$, which, by 
\theoremref{thm:bsp}, takes 
\begin{equation*}
\BO(h_i\Delta_i)=\BO\left(\frac{p_{i-1}\log^2 n}{p_i^2}\right)
=\BO\left(\left(\sqrt{n}\right)^{2^L/(2^L-1)}\log^2 n\right)
\end{equation*}
rounds. In addition, we need to relabel the nodes, which, as explained above,
can be done in time $\tilde\BO(h_i)\subseteq \tilde\BO(h_i\Delta_i)$, since the
depth of the shortest paths tree is bounded by $h_i\leq h_i\Delta_i$. Since
there are $L\leq \log \log n$ stages, the total number of rounds thus satisfies
the stated bounds. With respect to the label size, note that each stage $i$ adds
to the label of node $v$ the identifier of and distance to $\Lead_v(i)$ and a
tree label of size $(1+o(1))\log n$, for a total of $\BO(\log n)$ bits per
stage.
\end{proof}

\subsection{Long-Distance Routing}
\label{sec:skeleton}

We now explain how to route between the nodes in the top level of the
hierarchy created by the short-range scheme. Our central
concept is the \emph{skeleton graph},  defined as follows.

\begin{definition}[Skeleton Graph]
Let $G=(V,E,W)$ be a weighted graph. Given $S\subseteq V$ and $h\in \N$, the
\emph{$h$-hop skeleton-$S$ graph} is the weighted graph
$G_{S,h}=(S,E_{S,h},W_{S,h})$ defined by
\begin{compactitem}
\item $E_{S,h}\DEF \Set{\{v,w\}\mid v,w\in S,v\neq w,\mbox{and }\Hd(v,w)\le
h}$
\item For $\{v,w\}\in E_{S,h}$, define $W_{S,h}(v,w)$ to be the
$h$-weighted distance between $v$ and $w$ in $G$, i.e.,
$W_{S,h}(v,w)\DEF\Wd_h(v,w)$.
\end{compactitem}
\end{definition}

The main idea in the long-distance scheme is to construct a skeleton
graph with $S=S_L$ (the top level of the short-range hierarchy as constructed in
\sectionref{sec:short}). The choice of $h$ needs to balance two 
goals: on the one hand, the skeleton graph needs to accurately reflect the
distances of skeleton nodes in $G$, and on the other hand, we must be able to
quickly set up a tables that allow routing of small stretch between the skeleton
nodes.

A simple but crucial observation on skeleton graphs is that if the skeleton $S$
is a random set of nodes, and if $h\in \Omega(n\log n/|S|)$, then w.h.p., the
distances in $G_{S,h}$ are equal to the corresponding distances in $G$. This
means that it suffices to consider paths of $\BO(n\log n/|S|)$ hops in $G$ in
order to find the exact distances in $G$. The following lemma formalizes this
idea. (We state it for a skeleton \emph{containing} a random subset;
this generality will become useful in \sectionref{sec:steiner}.)

\begin{lemma}\label{lemma:distances}
Let $S_R$ be a set of random nodes defined
by $\Pr[v\in S_R]=\pi$ independently for all nodes for some given
$\pi$.
Let $S\supseteq S_R$.
If  $\pi\geq c\log n/h$ for a sufficiently large
constant $c>0$, then w.h.p., $\Wd_{S,h}(v,w)=\Wd(v,w)$ for all $v,w\in S$.
\end{lemma}
\begin{proof}
Fix $v,w\in S$. Clearly, $\Wd_{S,h}(v,w)\geq \Wd(v,w)$ because each
path in $G_{S,h}$ corresponds to a path of the same weight in $G$. We
need to show that $\Wd_{S,h}(v,w)\leq \Wd(v,w)$ as well. Let
$p=\Seq{u_0=v,u_1,\ldots,u_{\ell(p)}=w}$ be a shortest path
connecting $v$ and $w$ in $G$, i.e., $W(p)=\Wd(v,w)$.
We show, by induction on $\ell(p)$, that $\Wd(p)\leq\Wd_{S,h}(v,w)$
w.h.p.

For the basis of the induction note that if $\ell(p)\leq h$, then by definition
$\Wd_{S,h}(v,w)\leq W(p)=\Wd(v,w)$ and we are done. For the inductive step,
assume that the claim holds for all values of $\ell(p)\leq i$ for some $i\ge h$
and consider a path of length $\ell(p)=i+1$. Now, $\E[\,|S\cap
\Set{u_1,\ldots,u_{i}}\!|\,]\geq \E[\,|S_R\cap
\Set{u_1,\ldots,u_{i}}\!|\,]=i\pi\geq h\pi\in\Omega(\log n)$, and hence,
applying Chernoff's bound, we may conclude that w.h.p.\ the intersection is
non-empty. Let $u\in \{u_1,\ldots,u_i\}\cap S$. Since $p$ is a shortest path in
$G$, so are $(v,\ldots,u)$ and $(u,\ldots,w)$. Both these paths are of length at
most $i$, implying by the induction hypothesis that $\Wd_{S,h}(v,u)\leq\Wd(v,u)$
and $\Wd_{S,h}(u,w)\leq\Wd(u,w)$ w.h.p., respectively. Therefore
$\Wd_{S,h}(v,w)\leq \Wd_{S,h}(v,u)+\Wd_{S,h}(u,w)\leq
\Wd(v,u)+\Wd(u,w)=W(p)=\Wd(v,w)$, completing the induction. Note that the total
number of events we consider throughout the induction is bounded by a
polynomial in $n$, and since the probability of the bad events is polynomially
small, the union bound allows us to deduce that the claim holds w.h.p.
\end{proof}

Based on this observation, an obvious strategy to solve
long-distance routing is  to construct $G_{S,h}$ and 
compute its all-pairs shortest paths. But implementing this approach is not
straightforward. First, the edges of the skeleton graph are virtual: each
edge represents the shortest path of up to $h$ hops in $G$; and
second, the number of skeleton graph edges may be as large as
$\Omega(|S|^2)$. We solve both problems together: While computing the edges of
the skeleton graph, we sparsify the graph, bringing the number of edges
down to near-linear in the skeleton size. Once we are done, we can afford to let  
each skeleton node learn the full topology of the sparsified skeleton
graph, from which approximate all-pairs routes and distances can be computed
locally.

Technically, we use  the classical concept of sparse spanners, defined
as follows.

\begin{definition}[Weighted $k$-Spanners]
Let $H=(V,E,W)$ be a weighted graph and let $k\ge 1$. A weighted $k$-spanner of
$H$ is a weighted graph $H'=(V,E',W')$ where $E'\subseteq V$,
$W'(e)=W(e)$ for all $e\in E'$, and $\Wd_{H'}(u,v)\le k\cdot\Wd_H(u,v)$ for
all $u,v\in V$ (where $\Wd_H$ and $\Wd_{H'}$ denote weighted distances in $H$
and $H'$, respectively).
\end{definition}
We shall compute a spanner of the skeleton graph, while running
on the underlying physical graph, without ever
constructing the skeleton graph explicitly. We do this by
simulating the spanner construction algorithm of Baswana and
Sen~\cite{baswana07} on the implicit skeleton graph. Let us recall the
algorithm of~\cite{baswana07}; we use a slightly simpler variant that may select
some additional edges, albeit without affecting the probabilistic upper bound on
the number of spanner edges (cf.~\lemmaref{lem-reduce}). The input is a graph
$H=(V_H,E_H,W_H)$ and a parameter $k\in \N$.
\vspace*{1ex}\hrule
\begin{compactenum}
\item Initially, each node is a singleton \emph{cluster}:
$R_1:=\Set{\Set{v}\mid v\in V_H}$.
\item Repeat $k-1$ times (the $i^{th}$ iteration is called ``phase $i$''):
\begin{compactenum}
\item Each cluster from $R_i$ is \emph{marked} independently
with probability $|V_H|^{-1/k}$. $R_{i+1}$ is defined to be the set of
clusters marked in phase $i$.
\item \label{bs-first} If $v$ is a node in an unmarked cluster:
\begin{compactenum}
\item Define $Q_v$ to be the set of edges that consists of the lightest edge
from $v$ to each of the clusters $v\in R_i$ is adjacent to.
\item If $v$ has no adjacent marked cluster, then $v$ adds to the
spanner all edges in $Q_v$.
\item Otherwise, let $u$ be the closest neighbor of $v$ in a
marked cluster. In this case $v$ adds to the spanner the edge
$\{v,u\}$, and also all edges $\{v,w\}\in Q_v$ with $(W_H(v,w),w)<(W_H(v,u),u)$
(i.e., the identifiers $w,u$ break symmetry in case $W_H(v,w)=W_H(v,u))$.
Furthermore $v$ \emph{joins} the cluster of $u$ (i.e., if $u$ is in cluster $X$,
then $X:=X\cup \{v\}$).
\end{compactenum}
\end{compactenum}
\item \label{bs-final}
Each node $v$ adds, for each cluster $X\in R_k$ it is adjacent to, the lightest
edge connecting it to $X$.
\end{compactenum}\smallskip
\hrule\medskip
For this algorithm, Baswana and Sen prove the following result.
\begin{theorem}[\cite{baswana07}]
\label{thm-bs}
Given a weighted graph $H=(V_H,E_H,W_H)$ and an integer $k\ge 1$, the
algorithm above computes a $(2k-1)$-spanner of the graph. It has
$\BO(k|V_H|^{1+1/k}\log n)$ edges w.h.p.%
\footnote{In \cite{baswana07}, it is proved that the expected number of edges is
$\BO(k|V_H|^{1+1/k})$. The modified bound directly follows from
\lemmaref{lem-reduce}.}
\end{theorem}

\subsubsection*{Constructing the Skeleton Graph}

In our case, each edge considered in Steps \eqref{bs-first} and \eqref{bs-final}
of the spanner algorithm corresponds to a shortest path. Essentially, we
implement these steps in our setting by letting each skeleton node find its
closest $\BO(|S|^{1/k}\log n)$ clusters (w.h.p.)\ by running Algorithm $\BSP$.
We now explain how. First, all nodes $v$ in a cluster $X$ use the same source
identifier $\textrm{source}(v)=X$ (as if they were connected by a $0$-weight
edge to a virtual node $X$). This ensures that the overlap parameter needs to
account for the number of detected \emph{clusters} only, i.e., the number of
nodes per cluster is immaterial. Note that this implies that the plain version
of Algorithm $\BSP$ thus will not permit to determine to which node a skeleton
edge connects; hence we append to each communicated triple $(d,s,\Next)$ the
identifier of the actual endpoint $u\in SN(s)$ of the respective path and store
it when adding a corresponding triple to $L_v$ (without otherwise affecting the
algorithm). We refer to the modified algorithm as $\BSP'$. Second, regarding the
range parameter, \lemmaref{lemma:distances} shows that it is sufficient to
consider paths of $\BO(n\log n/|S|)$ hops only. Finally, the following lemma
implies that we may modify the spanner construction algorithm in a way that
allows us to use a small overlap parameter.

\begin{lemma}\label{lem-reduce}
  W.h.p., the execution of the centralized spanner construction
  algorithm yields identical results if in Steps~(\ref{bs-first})
  and~(\ref{bs-final}), each node considers the lightest edges to the
  $c|V_H|^{1/k}\log n$ closest clusters only (for a sufficiently large
  constant $c>0$).
\end{lemma}
\begin{proof}
Fix a node $v$ and a phase $1\le i<k$. If $v$ has at most $c|V_H|^{1/k}\log n$
adjacent clusters, the lemma is trivially true. So suppose that $v$ has more
than $c|V_H|^{1/k}\log n$ adjacent clusters. By the specification of
Step~(\ref{bs-first}), we are interested only in the clusters closer than the
closest marked cluster. Now, the probability that none of the closest
$c|V_H|^{1/k}\log n$ clusters is marked is $(1-|V_H|^{-1/k})^{c|V_H|^{1/k}\log
n}\in n^{-\Omega(c)}$. In other words, choosing a sufficiently large constant
$c$, we are guaranteed that w.h.p., at least one of the closest
$c|V_H|^{1/k}\log n$ clusters is marked. Regarding Step~(\ref{bs-final}),
observe that a cluster gets marked in all of the first $k-1$ iterations with
independent probability $|V_H|^{-(k-1)/k}$. By Chernoff's bound, the probability
that more than $c|V_H|^{1/k}\log n$ clusters remain in the last iteration is
thus bounded by $2^{-\Omega(c\log n)}=n^{-\Omega(c)}$. Therefore, w.h.p.\ no
node is adjacent to more than $c|V_H|^{1/k}\log n$ clusters in
Step~(\ref{bs-final}), and we are done.
\end{proof}
As a consequence of \lemmaref{lem-reduce}, we may invoke Algorithm $\BSP$ with
$\Delta\in\Theta(|S|^{1/k}\log n)$, and the time complexity of the invocation is
$\BO(|V|\cdot|S|^{-1+1/k}\log^2n)$. Detailed pseudo-code of our implementation
is given in \algref{algo:skeleton}. Each skeleton node $v\in S$ records the ID
of its cluster in phase $i$ as $F_i(v)$; nodes in $V\setminus S$ or those who do
not join a cluster in some round $i$ have $F_i(v)=\bot$. \algref{algo:edges} is
used as subroutine to implement Steps \eqref{bs-first} or \eqref{bs-final}
(Lines \ref{sk-first} or \ref{sk-final} of \algref{algo:skeleton},
respectively).

\begin{algorithm}[t!]
\small
\SetKwInOut{Input}{input}
\SetKwInOut{Output}{output}
\Input{
  $S$: set of skeleton nodes\\
  $k$: integer in $[1,\log n]$\REM{determines approximation
  ratio and number of spanner edges}
}
\Output{
  $E_{h,k}$: spanner edges of skeleton graph
  \REM{$h$ is defined in  \lineref{spanner-h}}\\
  $W_{h,k}:E_{h,k}\to \R^+$ \REM{weights of spanner edges}
}
$R_1:=\Set{\Set{w}\mid w\in S}$\REM{initial clusters are singletons of $S$}\\
\nllabel{st-bcast}
Broadcast $R_1$ to all nodes\\
\lForEach{$w\in V$}{\lIf{$w\in R_1$}{$F_1(w):=w$~}\lElse{$F_1(w):=\bot$}}
\REM{initializing leaders}\\
$h:=c\cdot n\log n/|S|$\REM{the constant $c$ controls the probability of
failure}
\label{spanner-h}\\
$\Delta:= c\cdot|S|^{1/k}\log n$\\
\For{$i:=1$ \emph{\KwTo} $k-1$}{
  $R_{i+1}:=$ uniformly random subset of $R_i$ of size
  $|S|^{1-i/k}=|R_i|/|S|^{1/k}$ \REM{select marked clusters}\nllabel{st-rand}\\
  Broadcast $R_{i+1}$ to all nodes\nllabel{st-bcast2}\\
  $(E(i),W(i)):=\text{edges}(F_i,R_{i+1},h,\Delta)$
  \REM{select spanner edges, phase $i$}\nllabel{sk-first}\\
  \ForEach{$w\in V$\nllabel{st-comp1}}{
    \If{$F_i(w)\in R_{i+1}$}{
      $F_{i+1}(w):=F_i(w)$
    }
    \Else{
      Let $E_w$ be the edges incident to $w$ in $E(i)$\\
      \If{$E_w\neq \emptyset$}{
        Let $\{w,u\}$ be the heaviest edge in $E_w$\\
        \lIf{$\text{marked}(u)$}{$F_{i+1}(w):=F_i(u)$}
      }
      \lElse{$F_{i+1}(w):=\bot$\nllabel{st-comp2}}
    }
    Broadcast $F_{i+1}$ to all nodes\nllabel{st-bcast3}
  }
}
$(E(k),W(k)):=\text{edges}(F_k,\emptyset,h,\Delta)$ \REM{final phase}
\nllabel{sk-final}\\
\lForEach{$e\in \bigcup_{i=1}^{k}E(i)$} {$W_{h,k}(e):=W(k)(e)$}\\
Broadcast $E_{h,k}:=\bigcup_{i=1}^{k}E(i)$ and $W_{h,k}$ to all
nodes\nllabel{sk-bcast}
\caption{Construction of long range routing skeleton at $v\in V$.}
\label{algo:skeleton}
\end{algorithm}

\begin{algorithm}[ht!]
\small
\SetKwInOut{Input}{input}
\SetKwInOut{Output}{output}
\Input{
  $F: V\to V\cup\{\bot\}$ \REM{locally known, $v$'s leader if $v$ is in a
  cluster, otherwise $\bot$}\\
  $R\subseteq V$ \REM{globally known, indicates (identifiers of leaders of)
  marked clusters}\\
  $h$ \REM{globally known, depth parameter of the search}\\
  $\Delta$ \REM{globally known, number of closest source clusters to detect}
}
\Output{
  $E_+$: edges added to the spanner\\
  $W_+:E_+\to \R^+$ edge weights
}
\ForEach{$w\in V$}{
  $\Src(w):=\left\{\begin{matrix}
    (F(v),0) & \text{if }F(v)\notin R\cup \{\bot\}\\
    (F(v),1) & \text{if }F(v)\in R\hfill\\
    \bot     & \text{else}\hfill
  \end{matrix}\right.\REM{distinguish marked from unmarked clusters}$
}
$L_v := \BSP'(h,\Delta,\Src)$\REM{variant of Algorithm $\BSP$ that keeps track
of path endpoints}\nllabel{ed-bsp}\\
$E_+:=\emptyset$\\
\If{$F_v\notin R\cup \{\bot\}$}{
  //\textit{for each entry $(\Wd,(f,b),u,w)\in L_v$, $u$ is the next hop on a
  path of weight $\Wd$ to $w$ in cluster $f$}\\
  $L_v:=L_v\setminus \Set{(0,(F(v),0),v,v)}$
  \REM{remove loops (clusters are in distance $0$ of themselves)}\\
  //\textit{recall that $L_v$ is ordered; first entry with
  $b=1$ corresponds to closest marked cluster}\\
  \ForEach{$(\Wd,(f,b),u,w)\in L_v$}{
    broadcast $(\Wd,\{v,w\})$ \REM{all nodes perform
      operation!}\nllabel{ed-bcast}\\
    $E_+:=E_+\cup \{v,w\}$ \nllabel{line:add}\\
    $W_+(\{v,w\}):=\Wd$\\
    \If{$f\in R$}{
      break \REM{$f$ is closest marked cluster}\nllabel{add-edges-end}
    }
  }
}
\Return$(E_+,W_+)$
\caption{$\text{edges}$: Edge detection and announcement for long
  range routing skeleton at $v\in V$.}\label{algo:edges}
\end{algorithm}

To prove the algorithm correct, we show that its executions can be mapped to
executions of the centralized algorithm, and then apply \theoremref{thm-bs}.
Below, we sketch the main points of such a mapping. The implementation of
\algref{algo:skeleton} is quite straightforward. Note that the broadcast steps
in Line \ref{st-bcast}, \ref{st-bcast2}, and \ref{st-bcast3} ensure that all
nodes know the clusters and which are the active clusters in each phase. The
random choices (\lineref{st-rand}) are made by cluster leaders, namely the nodes
$v$ for which $F_i(v)=v$. Lines \ref{st-comp1}--\ref{st-comp2} are local
computations each node does to get a global picture of the clusters
for the next phase. The correctness of the implementation of the edge selection
of Steps \ref{bs-first} and \ref{bs-final} of the centralized algorithm by 
\algref{algo:edges} was discussed above. We summarize with the following lemma.

\begin{lemma}\label{lemma:spanner}
Suppose the set $S$ input to \algref{algo:skeleton} contains a uniformly random
subset $S_R$ of $V$ and set $h(S_R)\DEF c \cdot n \log n/|S_R|$ for a
sufficiently large constant $c$. Then w.h.p.\ the following holds.
\begin{compactitem}
\item[(i)] \algref{algo:skeleton} computes a weighted $(2k-1)$-spanner of
the skeleton graph $G_{S,h(S_R)}$ that is known at all nodes and has
$\BO(|S|^{1+1/k}\log n)$ edges.
\item[(ii)] The weighted distances between nodes in $S$ are identical in
$G_{S,h(S_R)}$ and $G$.
\item[(iii)] The algorithm terminates in
$\tilde{\BO}(n/|S_R|^{1-1/k}+|S|^{1+1/k}+\HD)$ rounds.
\end{compactitem}
\end{lemma}
\begin{proof}
To prove Statement (i), we note that \algref{algo:skeleton} simulates
the centralized algorithm, except for considering only the closest
$\BO(|S|^{1/k}\log n)$ clusters in Lines \ref{sk-first} and
\ref{sk-final}. By \lemmaref{lem-reduce} and by \theoremref{thm:bsp},
this results in a (simulated) correct execution of the centralized algorithm
w.h.p. Hence Statement (i) follows from \theoremref{thm-bs}.

Regarding Statement (ii), observe that if $h(S_R)\geq n-1$, the statement holds by
definition since shortest paths cannot contain cycles and thus
$G_{S,h(S_R)}=G_{S,n-1}$. Otherwise, we have that $|S_R|\geq c\cdot \log n$,
implying by Chernoff's bound that w.h.p., the probability to select a node into
$S$ is $\pi\in \Theta(|S_R|/n)=\Theta((c\log n)/h(S_R))$. As by assumption $c$
is sufficiently large, Statement (ii) now follows from
\lemmaref{lemma:distances}.

For Statement (iii), consider first an invocation of \algref{algo:edges}. By
\theoremref{thm:bsp}, the invocation of \algref{alg:bsp} in \lineref{ed-bsp}
takes $\BO(\Delta h)\subset
\tilde{\BO}(|S|^{1/k}h(S_R))=\tilde{\BO}(n|S|^{1/k}/|S_R|)$ rounds.
The broadcast of \lineref{ed-bcast} is done globally. Each skeleton node may
communicate up to $\BO(|S|^{1/k}\log n)$ pieces of information for a total of
$\tilde{\BO}(|S|^{1+1/k})$ items. Doing this over a global BFS tree
takes $\tilde{\BO}(\HD+|S|^{1+1/k})$ rounds. As $k\leq \log n$, the total cost
of all invocations of \algref{algo:edges} is thus bounded by
$\tilde{\BO}(n|S|^{1/k}/|S_R|+|S|^{1+1/k}+\HD)$ rounds. Consider now
\algref{algo:skeleton}. The only non-local steps other than the invocations of
\algref{algo:edges} are the broadcasts, of which the most time consuming is the
one in \lineref{sk-bcast}, which takes $\BO(k|S|^{1+1/k}\log n+\HD)\subset
\tilde{\BO}(|S|^{1+1/k}+\HD)$ rounds.
\end{proof}


\subsubsection*{Routing on the Skeleton Graph}

\algref{algo:skeleton} constructs a $(2k-1)$-spanner of the skeleton graph and
made it known to all nodes. This enables each node to determine low-stretch
routing paths between any two skeleton nodes in $G_{S,h(S_R)}$ by local
computation. To use this information, we must be able, for each spanner edge
$\{s,t\}\in E_{S,h(S_R)}$, to route on a corresponding path in $G$, i.e., a path
of weight $W_{S,h(S_R)}(s,t)$. Since we rely on Algorithm~$\BSP$ during the
construction of the spanner, \theoremref{thm:bsp} shows that we can use the
computation to enable for each such edge to route from $s$ to $t$ \emph{or} from
$t$ to $s$: if, say, $s$ added the edge to the spanner, then following the
pointers computed during the execution of Algorithm~$\BSP$ yields a path of
weight $W_{S,h(S_R)}(s,t)$ from $s$ to $t$. However, in this case  $t$ might not
add $\{s,t\}$ to the spanner as well, and hence there is no guarantee that we
have sufficient information to route in \emph{both} directions.\footnote{Note
that unidirectionality is not an artifact of the specific implementation we
picked. E.g., in a star graph, the center has degree $n-1$, as it does in the
spanner. Hence we cannot expect the Bellmann-Ford pointers to give sufficient
information for bidirectional routing without further processing.} To resolve
this issue, we add a post-processing step where we ``reverse'' the
unidirectional routing paths, i.e., inform the nodes on the paths about their
predecessors. Note that this cannot be done in a purely local manner, as
exchanging the Bellmann-Ford routing pointers between neighbors will not tell a
node $s\in S$ which pointer to follow to reach a specific node $t\in S$ for
which $\{s,t\}$ is part of the spanner. However, \corollaryref{thm:bsp} states
that the (unidirectional) routing paths at our disposal have at most $h(S_R)$
hops. Taking into account that the spanner has few edges, it follows that
establishing bidirectional routing pointers can be performed sufficiently fast.

\begin{lemma}\label{lemma:spanner_routing}
Let $\{s,t\}$ be an edge of the spanner $G_{S,h(S_R)}$ that is selected by
\algref{algo:skeleton}.  W.h.p., after completing the algorithm, 
each node 
on the least-weight $s$--$t$ path of at most $h(S_R)$ hops in $G$  determine
the next hop on this path and the weight of the remaining subpath when routing
from $s$ to $t$ within $\tilde{\BO}(n/|S_R|+|S|^{1+1/k})$ rounds
\end{lemma}
\begin{proof}
For each edge $\{s,t\}$ added to the spanner by a node $s$, we route a message
on the shortest path of at most $h(S_R)$ hops from $s$ to $t$ in $G$. This
message initially contains the weight of the path, and each node on the path
subtracts the weight of the incoming path from this value. By
\theoremref{thm:bsp} this is feasible. When a node receives the message, it
records the immediate sender as the next hop on the path to $s$ and the weight
for future reference. By \lemmaref{lemma:spanner}, there are at most
$\tilde{\BO}(|S|^{1+1/k})$ edges in the constructed spanner of $G_{S,h(S_R)}$
w.h.p., implying that the maximal number of messages routed over each edge of
$G$ is bounded by $\tilde{\BO}(|S|^{1+1/k})$ w.h.p.\ as well. Moreover, no
routing path has more than $h(S_R)=c\cdot n \log n /|S_R|$ hops. Since the
messages traverse shortest $h$-hop paths, all of them reach their destinations
within the stated number of rounds~\cite{MP-91}.
\end{proof}

We now summarize the properties of the long-distance scheme. 
\begin{theorem}\label{theorem:spanner}
Suppose the set $S$ input to \algref{algo:skeleton} is a superset of a
uniformly random subset $S_R\subseteq V$ and $k\in \{1,\ldots,\log
n\}$. Then, w.h.p., 
within $\tilde{\BO}(n|S|^{1/k}/|S_R|+|S|^{1+1/k}+\HD)$ rounds, there are
routing tables for routing between nodes in $S$ with stretch $(2k-1)$.
\end{theorem}
\begin{proof}
Directly follows from Lemmas~\ref{lemma:spanner}
and~\ref{lemma:spanner_routing}.
\end{proof}

\subsection{Putting the Pieces Together}
\label{sec-tog}
Equipped with the results for the short-range and for the
long-distance routing, we can state the overall algorithm as a simple
composition of the two, linked by identifying the skeleton set from
the long-distance algorithm with the top level of the hierarchy $S_L$
of the short-range algorithm.  We run the long-range algorithm with
parameter $k$ to construct and make globally known the routing
skeleton and apply the short-range routing scheme with parameter $L$
to deal with nearby nodes.

Recall that the label of a node $w$ is
$\lambda(w)=\Seq{(\Lead_w(i),\Wd(w,\Lead_w(i)),\textrm{tree}_w(i)}_{i=0}^L$,
where $\textrm{tree}_w(i)$ denotes the label of $v$ in the tree on
$C_{\Lead_w(i)}$, and $\Lead_w(0)$ is simply $w$. Given the label $\lambda(w)$
to a node $v$, $v$ decides on the next routing hop as follows.
\begin{compactitem}
  \item If $\Lead_v(i)=\Lead_w(i)$ for some $i$, choose the next routing hop
  within $C_{\Lead_v(i)}$ to $w$ according to the respective tree label. In this
  case, $d$ is the distance from $v$ to $w$ in the tree (which can be computed
  from the distances of $v$ and $w$ to the root $\Lead_v(i)$ and whether the
  next routing hop is the parent of $v$ or a child).  
  \item Otherwise, node $v$ determines for each $i\in
  \{1,\ldots,L\}$ whether $\Lead_w(i-1)\in H_v(i)$. If so, it computes $d_i\DEF
  \Wd(v,\Lead_w(i-1))+\Wd(\Lead_w(i-1),w)$. Otherwise set $d_i\DEF \infty$.
  \item Next, denote by $S_v\subseteq S_L$ the set of skeleton nodes $v$ for
  which it stores a routing pointer and the corresponding path weight, and let
  for $s\in S_v$ $d_s$ be this weight. We define $\Wd^k$ to be the distance
  function on the spanner of the skeleton graph. Node $v$ computes
  $d_{L+1}\DEF\min_{s\in S_v}\{d_s+\Wd^k(s,\Lead_w(L))+\Wd(\Lead_w(L),w)\}$.
  \item Finally, $v$ computes $d\DEF \min_{i\in \{1,\ldots,L+1\}}\{d_i\}$, and
  determines the next routing hop in accordance with the corresponding path (ties
  broken by preferring smaller $i$), where we use the
  routing mechanism from \corollaryref{coro:bsp_route_stateless}.
\end{compactitem}
Since $v$ stores the tree routing tables for all trees on $C_{\Lead_v(i)}$,
the sets $H_i$ and the distances to the nodes in $H_i$, and
the complete spanner of the skeleton graph, together with the label $\lambda(w)$
it has the necessary information to perform all the above computations. Moreover, a next
routing hop is always determined, since $\Lead_v(L)\in S_v$ (by
\pprtyref{prop-y} of the short-range scheme) and therefore the set of
considered paths in the second step is non-empty. Finally, the routing decision
is stateless, since it depends on the local routing tables of $v$ and
$\lambda(w)$ only.

In order to show that indeed a route to $w$ of bounded stretch is determined by
the above routing decisions, we will show two properties. First, the value $d$
computed is the weight of a path of bounded stretch whose next routing hop $u$
is exactly the one computed by $v$, and second, the next node $u$ on the path
will compute a distance of at most $d-\Wd(v,u)$ to $w$. Since edge weights
are strictly positive, the latter immediately implies that the routes are
acyclic and will eventually reach their destination.
\begin{lemma}\label{lemma:d}
Fix any choice of the parameters $L$ and $k$ of the short range and long
distance schemes, respectively. For any node $v$ and label $\lambda(w)$,
consider the distance value $d$ and next routing hop $\Next$ computed by $v$
according to the above scheme. Then w.h.p., $d\leq (8kL-1)\Wd(v,w)$ and $\Next$
will compute a value $d'\leq d-W(v,\Next)$.
\end{lemma}
\begin{proof}
We show that $d\leq (8kL-1)\Wd(v,w)$ first. If $w\in\bigcup_{u\in
H_v(i)}C_u(i-1)$ for some $i\in \{1,\ldots,L\}$, observe that $d\leq
\Wd(v,\Lead_w(i))+\Wd(\Lead_w(i),v)$, and thus by \corollaryref{cor-short}
$d\leq (4L-3)\Wd(v,w)$. Otherwise, we have that $d=d_{L+1}$ since no other
routes are known to $v$. By definition and \pprtyref{prop-y} of the
short-range scheme, $d_{L+1}\leq
\Wd(v,\Lead_v(L))+\Wd^k(\Lead_v(L),\Lead_w(L))+\Wd(\Lead_w(L),w)$.
We bound
\begin{eqntext}
&&\Wd(v,\Lead_L(v))+\Wd^k(\Lead_v(L),\Lead_w(L))+\Wd(\Lead_L(w),w)&\\
&\le&\Wd(v,\Lead_L(v))+(2k-1)\Wd(\Lead_L(v),\Lead_L(w))+\Wd(\Lead_L(w),w)
& \text{by \theoremref{theorem:spanner}}\\
&\le& 2k\Wd(\Lead_L(v),v)+(2k-1)\Wd(v,w))+2k\Wd(w,\Lead_L(w)) & \text{triangle
inequality}\\
&\le& (2k(4L-1)+(2k-1))\Wd(v,w) &\text{\lemmaref{lem-sep}} \\
&=& (8kL-1) \Wd(v,w),
\end{eqntext}
proving that indeed $d\leq (8kL-1)\Wd(v,w)$.

Now let $\Next$ be the routing hop corresponding to $d$ computed by $v$. Due to
\pprtyref{prop-c} and \pprtyref{prop-h} of the short-range scheme, there are the
following three cases:
\begin{compactitem}
  \item $\Next$ is on the shortest path from $v$ to $\Lead_w(i-1)\in H_v(i)$ for
  some $i\in \{1,\ldots,L\}$ (this covers also the case that
  $\Lead_v(i-1)=\Lead_w(i-1)$, and in the tree on $C_{\Lead_w(i-1)}$ the
  connecting path traverses the root $\Lead_w(i-1)$);
  \item $\Next$ is on a path of weight $d_s$ to the node $s\in S_v$ minimizing
  the expression $d_s+\Wd^k(s,\Lead_w(L))+\Wd(\Lead_w(L),w)$;
  \item $\Next$ is on the shortest path from $\Lead_w(i)$ to $w$ for some $i\in
  \{1,\ldots,L\}$ (i.e., $\Lead_v(i)=\Lead_w(i)$, and in the tree on
  $C_{\Lead_w(i)}$ the connecting path does not traverse the root
  $\Lead_w(i)$).
\end{compactitem}
Regarding the first case, observe that since we are talking about shortest paths
in $G$ (not shortest $h$-hop paths), any source closer to $\Next$ than
$\Lead_w(i-1)$ will also be closer to $v$ than $\Lead_w(i-1)$. Hence
$\Lead_w(i-1)\in H_{\Next}(i)$. Since
$\Wd(\Next,\Lead_w(i-1))=\Wd(v,\Lead_w(i-1))-W(v,\Next)$, consequently $\Next$
will compute a distance of at most $d-W(v,\Next)$ to $w$.

In the second case, $\Next$ is either the next hop on a routing path as
constructed in \corollaryref{coro:bsp_route_stateless} or as constructed by the
``path reversal'' from \lemmaref{lemma:spanner_routing}. Either way, the
statements show that $\Next$ will know the next routing hop to $s$ as well as
the weight of the path. Since it knows the entire skeleton graph, it will thus
compute a distance of at most
$d_s-W(v,\Next)+\Wd^k(s,\Lead_w(L))+\Wd(\Lead_w(L),w)$ to $w$ as claimed.

For the third and final case, the statement trivially holds, since routing in
$C_{\Lead_w(i)}$ according to the tree routing table is on shortest paths and
will clearly lead to another node in $C_{\Lead_w(i)}$.
\end{proof}

It is fairly straightforward to set $k$ and $L$ to obtain a trade-off between 
the stretch of the routing scheme and the construction time. Specifically, we
can now state our main result as follows.
\begin{theorem}\label{thm-routing}
Let $1/2\le\alpha\le 1$ be given. Define $k\DEF\ceil{1/(2\alpha-1)}$ if
$\alpha\ge 1/2+1/\log n$, and $k\DEF\log n$ otherwise. Tables for stateless
routing and distance approximation with stretch
$\rho(\alpha)=8k\ceil{\log(k+1)}-1$ and label size $\BO(\log (k+1)\log n)$ can
be constructed in the \CONGEST\ model in $\tilde{\BO}(n^{\alpha}+\HD)$ rounds.
In particular, $\rho(1/2)\in \BO(\log n \log \log n)$ and $\rho(\alpha)\in
\BO(1)$ for any constant choice of $\alpha>1/2$.
\end{theorem}
\begin{proof}
A stretch bound of $8kL-1$ and the fact that the destination will indeed be
reached when following the computed pointers follows from \lemmaref{lemma:d}.
By \lemmaref{lem-short-perf}, the running time of the short-range construction
is bounded by $\tilde{\BO}((\sqrt{n})^{2^L/(2^L-1)})$ rounds w.h.p. The time
required for the skeleton construction is, by \theoremref{theorem:spanner},
$\tilde{\BO}(n/|S_L|^{1-1/k}+|S_L|^{1+1/k}+\HD) =
\tilde{\BO}((\sqrt{n})^{1+1/k}+\HD)$ w.h.p. To match the desired running time
bound of $\tilde\BO(n^{\alpha}+\HD)$ rounds, it thus suffices that
$\max\{1/k,1/(2^L-1)\}\geq 2\alpha-1-1/\log n$ (an additive $1/\log n$ in
the exponent maps to a constant factor). By choice of $k$,
this inequality holds for $L\DEF \ceil{\log(k+1)}$. The stretch
is thus bounded by
\begin{equation*}
\rho(\alpha)+1\leq 8kL=\left\{\begin{matrix}
8\lceil 1/(2\alpha-1)\rceil \lceil\log \lceil 1/(2\alpha-1)\rceil+1\rceil
& \text{for }\alpha>1/2\\
8\log n \lceil \log\log n +1\rceil & \text{for }\alpha=1/2.
\end{matrix}\right.
\end{equation*}

The bound on the label size follows from \lemmaref{lem-short-perf}, our
choice of $L$, and the fact that the long-distance scheme adds only $\BO(\log
n)$ bits to the label.
\end{proof}

The space complexity of our scheme, i.e., the number of bits of the computed
routing tables, is also straightforward to bound.

\begin{corollary}\label{coro:space}
The size of the routing table at node $v$ computed by the algorithm referenced
in \theoremref{thm-routing} is $\tilde{\BO}(n^{\alpha})$.
\end{corollary}
\begin{proof}
Observe that the dominant terms in memory consumption are (i) storing the sets
$H_v(i)$ and the next pointers to them for the short-range routing scheme, (ii)
storing the routing information for the paths from the roots $u\in S_i$ of the
trees induced by the sets $C_u(i)$, and (iii) storing $G_{S_L,h(S_L)}$ and the
next pointers for the long-range scheme. Trivially, the encoding of
$G_{S_L,h(S_L)}$ cannot require more than $\tilde{\BO}(n^{\alpha})$ memory, as
it is broadcasted globally over the BFS tree. The routing information from
$C_u(i)$ to the nodes in its tree is $\log^{\BO(1)}n$ bits~\cite{TZ-routing}. The
term from (i) originates from calls to Algorithm $\BSP$. The routing information
that needs to be stored consists of the history of the list maintained by
Algorithm $\BSP$. Hence, if such a call has depth and overlap parameters $h$ and
$\Delta$, the memory required is $\BO(h\Delta \log n)$. Hence the memory bound
for (i) directly follows from the running time bound from
\lemmaref{lem-short-perf}.
\end{proof}

\section{Extensions and Applications}
\label{sec-ext}

\subsection{Distance Sketches}

The problem of distributed distance sketches requires each node to
have a label and store a small amount of information (called the
\emph{sketch}), so that each node $v$ can estimate the distance to
each other node $u$ when given the label of $u$.%
\footnote{%
  The formulation in \cite{DDP} permits to use \emph{both} sketches to
  approximate the distance. However, from the distributed point of
  view it is more appropriate to assume that only a minimal amount of
  information is exchanged.}
Technically speaking, we already solved this problem, since our machinery
enables to estimate distances with small stretch. However, since $\alpha\geq
1/2$, the basic construction will always consume $\tilde\Omega(\sqrt{n})$
memory.

If we discard the routing information, we can reduce the space requirements of
the sketches at the expense of also increasing the stretch. To this end, we need
to reduce the maximal size of the sets $H_v(i)$ as well as the space consumed
for storing information on the skeleton graph. Our idea is as follows.
First, we change the sampling probabilities of the sets $S_i$ so that $|S_i|\in
\Theta(n^{1-i/L})$, where $i$ ranges from $0$ to $L$. With this choice, the
expected number of nodes from $S_i$ that are closer to a given node than the
closest node from $S_{i+1}$ is $\Theta(n^{1/L})$, implying that
$\mathbb{E}[|H_v(i)|]\in \Theta(n^{1/L})$ for all $i$. Second, we do not choose
$S_L$ as skeleton set, but rather $S_{i_0}$ for $i_0\approx L/2$, so that the
skeleton can be constructed quickly. To continue applying the short-range scheme
beyond stage $i_0$ without increasing the asymptotic time complexity of the
construction, we construct temporary distance sketches and labels for the
skeleton (using the long-range scheme with skeleton set $S_{i_0}$), which
allows nodes to estimate their distance to skeleton nodes locally. Ensuring that
the $S_i$ are subsets of the skeleton for $i\geq i_0$, each node thus can
simulate the short-range algorithm's detection of the sets $H_i$ and the
respective distance computation locally, based on its estimated distance to
skeleton nodes; the price we pay is increasing the stretch by factor $\BO(L)$
due to imprecise distances.

\begin{theorem}\label{theorem:sketches_small}
Given any integer $k\in [1,\ldots,\log n]$, distance sketches with
stretch $\rho(k)= 2k(8k-3)\in\BO(k^2)$, label size $\BO(k\log n)$,
and sketch size $\tilde{\BO}(n^{1/(2k)})$ can be constructed w.h.p.\ in
the \CONGEST\ model in $\tilde{\BO}(n^{1/2+1/(2k)}+\HD)$ rounds.
\end{theorem}
\begin{proof}
We use the following  algorithm, parametrized by $k$.
\begin{compactenum}
\item Run the short-range with $k$ stages and expected set sizes
$\E[|S_i|]=n^{i/(2k)}$ for $i\in \{0,\ldots,k\}$.
\item Run the long-range scheme on the skeleton $S_k$.
\item For $k+1\le i\le 2k-1$,  sample set $S_i\subseteq S_{i-1}$, where each node
is picked uniformly with probability $n^{-1/(2k)}$ in each step. Each node in
$S_k$ broadcasts its membership information.
\item For each pair $v\in V$ and $s\in S_k$, set
${\Wd'}(v,s)\DEF \Wd(v,\Lead_v(L))+\Wd(\Lead_v(L),s)$. For $k+1\le
i\le 2k-1$, compute at each node $v$ the closest node $\Lead_v(i)\in S_i$
w.r.t.\ $\Wd'$, and the set $H_v(i):=\{s\in
S_{i-1}\,|\,\Wd'(v,s)\leq \Wd'(v,\Lead_v(i))\}$. Set $H_v(2k)\DEF
S_{2k-1}$.
\item Store at each node $v$, for each $i\in \{1,\ldots,2k\}$: (i) the set
  $H_v(i)$; (ii) for 
each $u\in H_v(i)$, the value $\Wd(v,u)$ if $i\leq k$, or $\Wd'(v,u)$ if
$i>k$. Label node $v$ by $\lambda(v)=(\Lead_v(i),d(i))_{i\in
\{0,\ldots,2k-1\}}$, where $d(i)\DEF\Wd(v,\Lead_v(i))$ if $i<k$, and
$d(i)\DEF\Wd'(v,\Lead_v(i))$ otherwise. 
\end{compactenum}
Given label $\lambda(w)$ (which is clearly of size $\BO(k\log n)$), node $v$
estimates the distance to $w$ by finding the smallest $i\in \{1,\ldots,2k\}$ so
that $\Lead_w(i-1)\in H_v(i)$ and adding the respective distance estimates from
$v$ to $\Lead_w(i-1)$ (locally known) and from $\Lead_w(i-1)$ to $w$ (from the
label). Such an $i$ always exists because $H_v(2k)=S_{2k-1}$. This
completes the description of the algorithm.

By \corollaryref{cor-short}, the stretch of the routes represented by the
labels and sketches would be $8k-3\in\BO(k)$ w.h.p.\ if all distances were
exact. The approximation ratio is obtained by multiplying this value by the
maximal stretch of any distance estimates employed in the construction. Up to
stage $k$, all values are exact w.h.p. Thereafter, we use estimates of distances
between skeleton nodes and all other nodes. By the triangle
inequality, we have for all $v\in V$ and $s\in S$ that
\begin{equation*}
\Wd(v,s)\leq \Wd(v,\Lead_v(i_0))+\Wd(\Lead_v(i_0),s) \leq
\Wd(v,\Lead_v(i_0))+\Wd^k(\Lead_v(i_0),s)=\tilde\Wd(v,s).
\end{equation*}
On the other hand,
\begin{eqntext}
\Wd(v,\Lead_v(i_0))+\Wd^k(\Lead_v(i_0),s)&\leq & \Wd(v,s)+\Wd^k(\Lead_v(i_0),s)
& \text{by definition of } \Lead_v(i_0)\\
&\leq & 2k\cdot\Wd(\Lead_v(i_0),s) & \text{by \theoremref{theorem:spanner}}.
\end{eqntext}
Hence the stretch of the distance estimates is bounded by $\rho(k)=2k(8k-3)\in
\BO(k^2)$.

By Chernoff's bound, for all $i$ we have that $|S_i|\in \Theta(n^{(2k-i)/(2k)})$
w.h.p. Hence, the non-local part of the construction can be performed with
overlap parameter $\Delta_i\in \Theta(n^{1/(2k)}\log n)$ and distance parameter
$h_i\in \tilde\Theta(n^{i/(2k)})$ for all $i\in \{1,\ldots,k\}$. We conclude the
claimed running time and memory bounds of $\tilde\BO(n^{1/2+1/(2k)})$ (time
$\tilde\BO(h_i\Delta_i)$ for each step of the short-range scheme, time
$\tilde\BO(n/|S_k|^{1-1/k}+|S_k|^{1+1/k}+\HD)$ for the long-range scheme, and
time $\BO(|S_k|+\HD)$ for the additional broadcast step) and
$\tilde\BO(n^{1/(2k)})$, respectively, completing the proof.
\end{proof}

We note that a distributed implementation of Thorup-Zwick distance
oracles with stretch $2k-1$ and  running time $\tilde\Theta(\SPD\cdot
n^{1/k})$ was recently given by Das Sarma et al.~\cite{DDP}.
Intuitively, the reason for the discrepancy is that in \cite{DDP},
there is no use of the skeleton graph. In general, our running time and
the one from~\cite{DDP} are incomparable (one may run both algorithms in
parallel and use the output of the one that terminates first).

\subsection{Approximate Weighted Diameter}

Obtaining an approximation of the weighted diameter is simpler than constructing
distance sketches. Dropping the short-range scheme from the construction, we can
prove the following result.
\begin{theorem}\label{thm-diameter}
For any $k\in \N$, the weighted diameter $\WD$ can be approximated w.h.p.\ to
within a factor of $2k+1$ in the \CONGEST\ model in
$\tilde{\BO}(n^{1/2+1/(2k)}+\HD)$ rounds.
\end{theorem}
\begin{proof}
We use the following streamlined version of our algorithm. 
\begin{compactenum}
\item Select a uniformly random skeleton $S$ where
  $\Pr[v\in S]=1/\sqrt n$ independently for all $v\in V$.
\item \label{stD-1} Apply \algref{algo:skeleton} to construct a
  $(2k-1)$-spanner of $G_{S,h(S_R)}$. Let $\WD^k$ be the weighted
  diameter of the spanner of $G_{S,h(S_R)}$ (which can be computed
  locally).
\item \label{stD-2} Apply Algorithm $\BSP$, where
  all nodes in $S$ function as the same source: $\Src(v)=1$ for all
  $v\in S$ and $\Src(v)=\bot$ for all $v\notin S$.
Use $\Delta=1$ and  $h\in \Theta(\sqrt{n}\log n)$.
\item \label{stD-3} Find the maximal distance $d_{\max}$ computed by
  any node and output $2d_{\max}+\WD^k$.
\end{compactenum}
Regarding the time complexity, note that Step \ref{stD-1} requires
$\tilde{\BO}(n^{1/2+1/(2k)}+\HD)$ rounds by
\theoremref{theorem:spanner},
Step \ref{stD-2} requires $\tilde{\BO}(\sqrt n)$ time by
\theoremref{thm:bsp}, and Step \ref{stD-3} takes $\BO(\HD)$ rounds. 

Regarding  the approximation ratio,  consider any $v,w\in V$, and let
$s_v,s_w\in S$ the nodes in $S$ closest to $v$ and $w$,
respectively. Then
\begin{equation*}
\Wd(v,w)\leq \Wd(v,s_v)+\Wd(s_v,s_w)+\Wd(s_w,w)\leq d_{\max}+\WD^k+d_{\max}~,
\end{equation*}
and hence $\WD\le\WD^k+2d_{\max}$. 
On the other hand, we have 
\begin{eqnarray*}
\WD&\geq & \max\left\{d_{\max},\max_{s,t\in S}\{W_{S,h(S_R)}(s,t)\}\right\}\\
&\geq & \frac{2d_{\max}+(2k-1)\max_{s,t\in S}\{W_{S,h(S_R)}(s,t)\}}{2k+1}
~\geq ~ \frac{2d_{\max}+\WD^k}{2k+1}~,
\end{eqnarray*}
where the last inequality holds w.h.p.\ by
\theoremref{theorem:spanner}. 
\end{proof}

\subsection{Distributed Approximation for Generalized Steiner Forest}
\label{sec:steiner}

In this section we explain how to utilize our routing scheme to obtain
a fast distributed algorithm for the Generalized Steiner Forest
problem (\gsf), defined as follows.

\begin{quote}\textbf{Generalized Steiner Forest} (\gsf')\\
  \textbf{Input:} A weighted graph $G=(V,E,W)$, a set of \emph{terminals}
  $T\subseteq V$, and for each terminal $t\in T$ a component number $C(t)$.\\
  \textbf{Output:} A subset of the edges $F\subseteq E$ such that for all pairs
  $s,t\in T$ with $C(s)=C(t)$, we have that $s$ is connected to $t$
  in the graph   $(V,F)$.\\
  \textbf{Goal:} Minimize $\sum_{e\in F}W(e)$.
\end{quote}
We note that sometimes, the connectivity requirement $C$ is expressed
as a set of node pairs. While the size of the input representation may differ,
the two variants are equivalent for our purposes; using our spanner
construction, we can obtain the component-based description within
$\tilde\BO(T+\HD)$ rounds from the pair-based formulation.

In the distributed setting, we assume that each node knows whether it is a
terminal, and if so, what is its component number. Clearly, we can establish
global knowledge on $T$ and the component numbers within time $\BO(T+\HD)$ by
broadcasting the respective pairs of values over a BFS tree. 

We now present a solution to \gsf. We start with a generic reduction to a
centralized algorithm which abstracts away the underlying graph $G$,
and uses a graph whose nodes are just the terminals and edge weights
are inter-terminal distance estimates.

\begin{algorithm}[ht!]
\small
\SetKwInOut{Input}{input}
\SetKwInOut{Output}{output}
\Input{
  terminal components \REM{locally known: $v$ knows whether $v\in T$ and if
  so, its component number}\\
}
\Output{
  $F$: edges in the Steiner forest
}
Obtain distance estimates for all distances $\Wd(s,t)$ with
  $s,t\in T$.\label{gsf-1}\\
Simulate $\ALG$ on the graph $G'=(T,E',W')$ and the same terminal components,
  where $E':=\Set{\{s,t\}\,|\, s,t\in T\wedge s\neq t}$, and for all $s,t\in T$, 
  $W'(s,t):=\min\{\tilde{\Wd}_s(t),\tilde{\Wd}_t(s)\}$ and $\tilde{\Wd}_s(t)$
  denotes the estimate of $\Wd(s,t)$ computed at $s$ (given the label of
  $t$). Denote by $F'$ the computed solution.\label{gsf-2}\\
Identify and output all edges on paths in $G$ that correspond to edges in $F'$.
\label{gsf-3}
\caption{Distributed algorithm for \gsf. $\ALG$ is any centralized
approximation algorithm for \gsf'.}\label{algo:gsf}
\end{algorithm}
To analyze \algref{algo:gsf}, we consider two simple transformations of
the input instance and state their effect on the cost of the solution.
First, consider the effect of using just distances between terminals (and not
the whole graph). The following lemma bounds the effect of this 
simplification. Given an instance $\cI$ for \gsf, we use $\OPT(\cI)$ to denote
any fixed optimal solution for $\cI$.
\begin{lemma}\label{lem-terminals-only}
Let $\cI=\left(G=(V,E,W),T,C\right)$ be an instance of \gsf.
Define an instance
$\cI'=\left(G',T,C\right)$, where $G'=(T,E',W')$ with
$E'=\Set{\{s,t\}\,|\, s,t\in T}$ 
and $W'(s,t)=\Wd_G(s,t)$. Then $W'(\OPT(\cI'))\le 2W(\OPT(\cI))$.
\end{lemma}
\begin{proof}
The proof is a generalization of the standard argument for Steiner
trees \cite{TM-80}.
Let $F=\OPT(\cI)$. Let $C_1=(V_1,F_1),\ldots,C_m=(V_m,E_m)$ be the connected
components of $F$. By optimality of $F$, the  components are trees.
Fix a component $C_j$, and consider an Euler tour of
its tree. Let $\sigma=\Seq{v_0,\ldots,v_{2|V_j|}=v_0}$ be the sequence of
nodes visited by the tour. Since each edge in $F_j$ is visited
exactly twice in $\sigma$, we have that the total weight of edges in
the tour is $2W(F_j)$. Define the node
sequence $\sigma'=\Seq{u_0,\ldots,u_{|V_j|-1}}$ obtained from $\sigma$
by omitting the second occurrences of
nodes from $\sigma$. Consider now the set of edges
$F'_j:=\Set{\{u_{i-1},u_i\}\mid 0<i<|V_j|}$ in $G'$. Since the edges in $G'$
are shortest paths in $G$, clearly the weight of $F'_j$ is not more than
the total weight of edges in $\sigma$, namely $W'(F'_j)\le
2W(F_j)$.
Finally, note that
$F'=\bigcup_{j=1}^tF'_j$ is a feasible solution for $\cI'$, and
therefore
$$
W'(\OPT(\cI'))
~\le~W'(F')
~=~\sum_{j=1}^tW'(F'_j)
~\le~\sum_{j=1}^t2W(F_j)
~=~2W(\OPT(\cI))~.
$$
\end{proof}
Next, consider replacing edge weights by $\rho$-approximate weights.
\begin{lemma}\label{lem-gsf-appx}
Let $\cI=\left(G=(V,E,W),T,C\right)$ and
$\cI'=\left(G=(V,E,W'),T,C\right)$ be instances of \gsf\
differing only in the edge weights as follows: for all $e\in E$,
$W(e)\le W'(e)\le \rho W(e)$ for some $\rho\ge1$. Then $W(\OPT(\cI'))\le \rho
W(\OPT(\cI))$.
\end{lemma}
\begin{proof}
$W(\OPT(\cI'))\leq W'(\OPT(\cI'))\leq W'(\OPT(\cI))\leq \rho W(\OPT(\cI))$.
\end{proof}

The preceding two lemmas show that using $\rho$-approximate distances and a
centralized $a$-approxima\-tion algorithm for $\gsf$, we will obtain a
distributed $(2\rho a)$-approximation algorithm for $\gsf$.

It remains to show how to efficiently implement \algref{algo:gsf} in the
\CONGEST\ model. The key is Step \ref{gsf-1}: Steps \ref{gsf-2} and
\ref{gsf-3} will be performed locally at each node.
\begin{corollary}\label{coro:gsf_impl}
For any integer $k\in [1,\log n]$, \algref{algo:gsf} can be
executed in the \CONGEST\ model in
$\tilde{\BO}((\sqrt{n}+|T|)^{1+1/k}+\HD)$ rounds with stretch factor 
$\rho(k)=2k-1$.
\end{corollary}
\begin{proof}
We apply the long-range routing scheme with skeleton set $S:=T\cup R_S$, where
$R_S$ is sampled uniformly and independently at random with probability
$n^{-1/2}$ from $V$. \lemmaref{lemma:spanner} implies that we can perform
Step~\ref{gsf-1} of \algref{algo:gsf} within
$\tilde{\BO}(n|S|^{1/k}/|S_R|+|S|^{1+1/k}+\HD)
=\tilde{\BO}((\sqrt{n}+|T|)^{1+1/k}+\HD)$ rounds with stretch $\rho(k)=2k-1$.
Moreover, at the end of this step, all nodes know the spanner of the skeleton
graph and can therefore locally compute $W'$. As remarked earlier, all nodes can
learn the terminal components within $\tilde{\BO}(\sqrt{n}+\HD)\subset
\tilde{\BO}((\sqrt{n}+|T|)^{1+1/k}+\HD)$ rounds. With this information in place,
all nodes can locally simulate $\ALG$ on $G'$ and thus perform Step~\ref{gsf-2}
of the algorithm. According to \lemmaref{lemma:spanner_routing}, the nodes on
paths in $G$ corresponding to edges in $G'$ can learn of their membership within
$\tilde{\BO}(n|S|^{1/k}/|S_R|+|S|^{1+1/k}+\HD)$ rounds as well. Afterwards,
Step~\ref{gsf-3} of the algorithm can be completed locally as well. Summing up
the running time bounds for the individual steps, we conclude that the overall
time complexity is $\tilde{\BO}((\sqrt{n}+|T|)^{1+1/k}+\HD)$ as claimed.
\end{proof}

Altogether, we arrive at the following result.
\begin{theorem}
Given any integer $k\in [1,\log n]$ and any centralized $a$-approximation
algorithm to \gsf, \gsf\ can be solved in the \CONGEST\ model with approximation
ratio $2a(2k-1)$ in $\tilde{\BO}((\sqrt{n}+|T|)^{1+1/k}+\HD)$ rounds, where $T$
denotes the set of terminal nodes.
\end{theorem}
\begin{proof}
\corollaryref{coro:gsf_impl} proves that \algref{algo:gsf} can be implemented in
$\tilde{\BO}((\sqrt{n}+|T|)^{1+1/k}+\HD)$ rounds, with distance estimates of
stretch $\rho(k)=2k-1$.
The approximation guarantee therefore follows from
Lemmas~\ref{lem-terminals-only} 
and~\ref{lem-gsf-appx}.
\end{proof}

We note that one can implement Step \ref{gsf-1} also by computing distance
sketches as in \theoremref{thm-routing} without including the entire terminal
set into the skeleton (i.e., $G'$ does not become global knowledge), and
simulate $\ALG$ sequentially, taking $\BO(\HD)$ per step. This will reduce the
overall running time in case $|T|\gg\sqrt{n}$ and $\HD$ times the step
complexity of $\ALG$ (in terms of the number of globally synchronized steps) is
small compared to $|T|^{1+1/k}$. However, this approach has two drawbacks. 
First, $\ALG$ cannot be arbitrary, but must admit to be simulated via a BFS tree
using small messages. Second, the approximation ratio deteriorates according to
the number of stages of the short-range scheme, as the number of stages
contributes as a multiplicative factor $\rho$.

\noindent\textbf{Discussion.} It is known that the special case of MST, where
all nodes are terminals in a single component has worst-case running time of
$\tilde\Omega(\sqrt n)$ even if the hop-diameter is $\BO(\log n)$
\cite{PelegR-00}. However, it is unclear whether this lower bound holds if the
number of terminals is small, and in turn, whether a larger number of terminal
components makes the problem harder. For instance, for a single pair of
terminals the problem reduces to selecting a single approximate shortest path;
we are not aware of any non-trivial lower bound on this problem. Khan et
al.~\cite{KKMPT} provide a $\BO(\log n)$-approximation to GSF within
$\tilde\BO(\SPD\cdot \gamma)$ rounds, where $\gamma$ denotes the number of
terminal components. The algorithm from~\cite{KKMPT} matches this bound up to factor
$\log^{\BO(1)}n$ if $\SPD\cdot\gamma\in \tilde\BO(\sqrt{n})$; our approach does
so in case $t\in \tilde\BO(\sqrt{n})$, where $t$ is the number of terminal nodes. Note
that the two running time bounds in general are incomparable: for approximation
ratio $\BO(\log n)$, we achieve time complexity $\tilde\BO(\sqrt{n}+t+\HD)$;
there are instances for which $\SPD\cdot \gamma \gg \sqrt{n}+t$ as well as those
where $\sqrt{n}+t \gg \SPD \cdot \gamma$. However, our approach is superior in
that we can, for any integer $k\in [1,\log n]$, ensure an approximation ratio of
$\BO(k)$ at the expense of a slightly larger running time of
$\tilde\BO((\sqrt{n}+t)^{1+1/k}+\HD)$ rounds. The authors of~\cite{KKMPT} employ
probabilistic tree embeddings, a technique for which an approximation ratio of
$\Omega(\log n)$ is inherent \cite{FRK}.

\subsection{Tight Labels}

The presented routing scheme relabels the nodes according to the Voronoi
partition on each level. This yields suboptimal size of labels and makes it
impossible for nodes to learn all labels $\lambda(V)$ quickly. We now
present a modification of our routing scheme with labels
$\lambda(V)=\{1,\ldots,n\}$, trading in a larger stretch.

Instead of labeling the nodes on each level of the hierarchy independently, we
do this by an inductive construction.
\begin{compactenum}
\item Define the partial order $\prec$ on $V$ given by $v\prec u$ if (and only
if) one of the following is true:
\begin{compactitem}
\item $v,u\in S_L$ and the identifier of $v$ is smaller than the identifier of
$u$.
\item $l_v=l_u<L$, $\Lead_v(l_v+1)=\Lead_u(l_v+1)$, and $v$ precedes $u$ in a
fixed DFS enumeration of the tree $T_{\Lead_v(l_v+1)}(l_v+1)$ on
$C_{\Lead_v(l_v+1)}(l_v+1)$ induced by the shortest $h_{l_v}$-hop paths from
each $w\in C_{\Lead_v(l_v+1)}(l_v+1)$ to $\Lead_v(l_v+1)$ detected by the
invocation of Algorithm~$\BSP$ in stage $l_v$.
\end{compactitem}
\item \label{label-2} Set $\mathrm{count}_v(0):=1$ for all $v\in V=S_0$. For
each level $i\in \{1,\ldots,L\}$, aggregate the sums of the values
$\mathrm{count}_s(i-1)$ of nodes $s\in S_{i-1}$ in subtrees of
$T_{\Lead_s(i)}(i)$ at the roots of these subtrees. We define for all $s\in S_i$
the value
\begin{equation*}
\mathrm{count}_s(i):=\sum_{\substack{s'\in S_{i-1}\\
s=\Lead_{s'}(i)}}\mathrm{count}_{s'}(i-1),
\end{equation*}
which can be computed from the received values. For each level $i$, this
operation can be performed within $\BO(h_i)$ rounds.
\item Each skeleton node $s\in S_L$ announces $\mathrm{count}_s(L)$ to all
other nodes. This requires $\BO(|S_L|+\HD)$ rounds using a BFS tree.
\item Each skeleton node $s\in S_L$ sets 
\begin{equation*}
\lambda(s):=1+\sum_{\substack{s'\in S_L\\ s'\prec s}}\mathrm{count}_{s'}(L).
\end{equation*}
\item Starting from level $L$ and proceeding inductively on decreasing $i$, for
each $i\in \{1,\ldots,L\}$, each node $s\in S_i$, and each node $s'\in
C_s(i)\cap S_{i-1}$, we inform $s'$ of the value
\begin{equation*}
\lambda(s'):=\lambda(s)+1+\sum_{\substack{s''\in C_s(i)\cap S_{i-1}\\
s''\prec s'}}\mathrm{count}_{s''}(i-1).
\end{equation*}
Note that this step can be performed in $\BO(h_i)$ rounds once $\lambda_s$ is
known due to the information collected in Step~\ref{label-2}.
\end{compactenum}
From the above arguments and the results from \sectionref{sec:short} we can
immediately conclude that the time complexity of computing these labels is
negligible.
\begin{corollary}\label{coro:label_time}
Executing the above construction does not increase the asymptotic time
complexity of setting up the routing tables.
\end{corollary}

By construction, $\lambda(V)=\{1,\ldots,n\}$. Note that at the end of the above
construction, each node $v$ knows $\lambda(v)$ and, for each level $i$ and each
of its children in $T_{\Lead_v(i)}(i)$, it knows the range of labels associated
with this child. For each $v\in V$, set $v_0:=v$ and define inductively for
$i\in \{1,\ldots,L\}$ that $v_i=\Lead_{v_{i-1}}(i)$. Given any label
$\lambda(w)$, we can thus route from any node $v$ to $w$ as follows.
\begin{compactenum}
\item Set $i:=0$.
\item If $i<L$ and $w_i\notin H_{v_i}(i+1)$, then route to $v_{i+1}$, set
$i:=i+1$, and repeat this step. If $i<L$ and $w_i\in H_{v_i}(i+1)$, route to
$w_i$ and proceed to the next step. If $i=L$, route to $w_L$ using the
long-range scheme and proceed to the next step.
\item If $i>0$, then route to $w_{i-1}$, set $i:=i-1$, and repeat this step.
Otherwise $w_i=w$ and we are done.
\end{compactenum}
The constructed sequence of routing indirections is thus
$(v_0,\ldots,v_{i_0},w_{i_0},\ldots,w)$, where either $i_0$ is the minimal
level such that $w_{i_0}\in H_{v_{i_0}}(i_0+1)$ or $i_0=L$. As a result of these
indirections, we cannot give a bound on the stretch that is linear in the
number of levels anymore. However, we still can argue that $\Wd(v_i,w_i)\leq
4\Wd(v_{i-1},w_{i-1})$ assuming that $w_{i-1}\not \in H_{v_{i-1}}(i)$.
\begin{lemma}\label{lem-stretch-unique}
Suppose that for the labeling scheme stated above we have for some integer
$1\leq i_0\leq L$ that $w_{i-1}\not \in H_{v_i}(i)$ for all integers $0\leq
i<i_0$.
Then
\begin{equation*}
\Wd(v_{i_0},w_{i_0})+\sum_{i=1}^{i_0}(\Wd(v_{i-1},v_i)+\Wd(w_i,w_{i-1}))
< 2\cdot 4^{i_0}\Wd(v,w).
\end{equation*}
\end{lemma}
\begin{proof}
We show by induction that $\Wd(v_i,w_i)\leq 4^i\Wd(v,w)$ for all $0\leq i\leq
i_0$, which is obviously true for $i=0$. Analogously to \lemmaref{lem-sep} we
have that $\Wd(v_i,v_{i+1})\leq \Wd(v_i,w_i)$ and $\Wd(w_i,w_{i+1})\leq
2\Wd(v_i,w_i)$. By the triangle inequality,
\begin{equation*}
\Wd(v_{i+1},w_{i+1})\leq \Wd(v_{i+1},v_i)+\Wd(v_i,w_i)+\Wd(w_i,w_{i+1})\leq
4\Wd(v_i,w_i)=4^{i+1}\Wd(v,w).
\end{equation*}
This completes the induction and in addition reveals that
\begin{equation*}
\sum_{i=1}^{i_0}(\Wd(v_{i-1},v_i)+\Wd(w_i,w_{i-1}))\leq
3\sum_{i=0}^{i_0-1}\Wd(v_i,w_i)\leq 3\sum_{i=0}^{i_0-1}4^i\Wd(v,w)
<4\cdot 4^{i_0-1}\Wd(v,w).
\end{equation*}
Therefore
\begin{eqnarray*}
\Wd(v_{i_0},w_{i_0})+\sum_{i=1}^{i_0}(\Wd(v_{i-1},v_i)+\Wd(w_i,w_{i-1}))
&<& 2\cdot 4^{i_0}\Wd(v,w),
\end{eqnarray*}
concluding the proof.
\end{proof}
As by \corollaryref{coro:label_time} the construction time of the routing scheme 
is not affected by the above labeling and routing mechanism and
\lemmaref{lem-stretch-unique} provides a stretch bound of $2\cdot 2^{2L}$ for
the modified short-range routes, we obtain the following statement.
\begin{theorem}\label{theorem:routing-unique}
Given $\alpha\in [1/2,1]$, let $k=\lceil 1/(2\alpha-1)\rceil$ if
$\alpha>1/2+1/\log n$ and $k=\lfloor 1/\log n\rfloor$ otherwise. Tables for
stateless routing and distance approximation with stretch $\rho(\alpha)=4k\cdot
4^{\lceil\log (k+1)\rceil}+2k-1\in \BO(k^3)$ with node labels $1,\ldots,n$ can
be constructed in the \CONGEST\ model in $\tilde{\BO}(n^{\alpha}+\HD)$ rounds.
\end{theorem}

\section*{Acknowledgements}
We would like to thank David Peleg for valuable discussions.

\bibliographystyle{abbrv}
\bibliography{distance}

\end{document}